\pdfoutput=1
\documentclass[letterpaper, 10 pt, conference]{ieeeconf}
\IEEEoverridecommandlockouts
\let\labelindent\relax
\usepackage{amsmath,amssymb,amsthm,mathtools}
\usepackage{cite}
\usepackage{enumitem}
\usepackage{tikz}
\usepackage{float}
\usetikzlibrary{positioning,arrows.meta}
\makeatletter\let\NAT@parse\undefined\makeatother
\usepackage[normalem]{ulem}
\usepackage[colorlinks=true,linkcolor=blue,citecolor=blue]{hyperref}
\usepackage[capitalize,nameinlink]{cleveref}

\newtheorem{theorem}{Theorem}
\newtheorem{proposition}{Proposition}
\newtheorem{lemma}{Lemma}
\theoremstyle{definition}
\newtheorem{definition}{Definition}
\newtheorem{assumption}{Assumption}
\newtheorem{problem}{Problem}
\theoremstyle{remark}
\newtheorem{remark}{Remark}
\crefname{problem}{Problem}{Problems}
\crefname{assumption}{Assumption}{Assumptions}

\newcommand{\R}{\mathbb{R}}
\newcommand{\xb}{\bar{x}}
\newcommand{\nb}{n_p+m}
\newcommand{\Ab}{\bar{A}}
\newcommand{\By}{B_y}
\newcommand{\Clim}{C_{\lim}}
\newcommand{\Cz}{C_z}
\newcommand{\Cbar}{\overline{C}}
\newcommand{\Kfull}{K_{\mathrm{full}}}
\newcommand{\Hu}{H_u}
\newcommand{\ycmd}{y_{\mathrm{cmd}}}
\newcommand{\yg}{y_g}
\newcommand{\ag}{\alpha_g}
\newcommand{\zmax}{z_{\max}}
\newcommand{\zmin}{z_{\min}}
\newcommand{\umax}{u_{\max}}
\newcommand{\umin}{u_{\min}}
\newcommand{\Z}{\mathcal{Z}}
\newcommand{\U}{\mathcal{U}}
\newcommand{\W}{\mathcal{W}}
\newcommand{\Sset}{S}
\newcommand{\SE}{\Sset_{E}}
\newcommand{\Oh}{\mathcal{O}}
\newcommand{\Nt}{\widetilde{\mathcal{N}}}
\newcommand{\Creg}{C_{\mathrm{reg}}}
\newcommand{\Dreg}{D_{\mathrm{reg}}}
\newcommand{\Bz}{B_\zeta}
\newcommand{\etat}{\tilde{\eta}}
\newcommand{\Wt}{\widetilde{W}}
\newcommand{\Bd}{B_d}
\newcommand{\Tz}{T_{\xi}}
\newcommand{\Emin}{E_{\min}}
\newcommand{\ThB}{\bar{\Theta}}
\newcommand{\ThH}{\hat{\Theta}}
\newcommand{\nub}{\bar{\nu}}
\newcommand{\eM}{e_{M}}
\newcommand{\ones}{\mathbf{1}}
\newcommand{\T}{^{\top}}
\newcommand{\Pieta}{\Pi_{\eta}}
\newcommand{\Imu}{\mathcal{I}_{\mu}}
\newcommand{\Isig}{\mathcal{I}_{\sigma}}
\newcommand{\Izero}{\mathcal{I}_{0}}
\newcommand{\Kz}{\mathcal{K}_{0}}
\newcommand{\Kp}{\mathcal{K}_{+}}
\newcommand{\arxref}[1]{\eqref{#1}}
\DeclareMathOperator*{\argmin}{arg\,min}

\usepackage{algorithm}
\usepackage{algorithmic}
\usepackage[most]{tcolorbox}
\usepackage{colortbl}
\graphicspath{{figs/}}
\newif\ifcolorboxes
\colorboxestrue
\ifcolorboxes\colorlet{keyrow}{black!8}\else\colorlet{keyrow}{white}\fi
\newcommand{\stmtbox}[2]{\ifcolorboxes\tcolorboxenvironment{#1}{enhanced jigsaw,
breakable, frame hidden, boxrule=0pt, sharp corners, colback=#2,
left=3pt, right=3pt, top=3pt, bottom=3pt, boxsep=0pt,
grow to left by=0pt, grow to right by=0pt,
before skip=8pt, after skip=8pt}\fi}
\stmtbox{definition}{blue!7}
\stmtbox{assumption}{teal!10}
\stmtbox{lemma}{yellow!13}
\stmtbox{proposition}{green!9}
\stmtbox{theorem}{red!8}
\stmtbox{problem}{orange!10}
\stmtbox{remark}{black!5}

\begin{document}

\title{\LARGE \bf Feasibility of Simultaneous Input-Output Constraints for Tracking in a Class of LTI Systems: Part~I}

\author{Junhyeok~Yoon$^{1}$, Heather~Hussain$^{2}$, and Anuradha~M.~Annaswamy$^{1}$\thanks{$^{1}$J. Yoon and A. M. Annaswamy are with the Massachusetts Institute of Technology, Cambridge, MA 02139, USA (e-mail: yjunh01@mit.edu; aanna@mit.edu).}\thanks{$^{2}$H. Hussain is with The Boeing Company, Engineering \& Technology Innovation, USA (e-mail: heather.s.hussain@boeing.com).}\thanks{This work was supported by the Boeing Strategic University Initiative.}}

\maketitle

\begin{abstract}
This paper addresses simultaneous input and output constraint satisfaction for a class of multi-input linear time-invariant systems with state feedback and integral action. A control barrier function (CBF)-based governor modifies the reference command while preserving the nominal feedback controller. Necessary and sufficient conditions are derived under which the governor is feasible at every state of a prescribed operating set and ensures simultaneous input and output constraint satisfaction, forward invariance, and bounded closed-loop solutions. The paper also shows that these conditions need not be met for certain choices of the CBF governor parameters due to conflicts between input and output constraints. A systematic design procedure is proposed that first searches for feasible free parameters and, if necessary, proposes a relaxation in the input constraint that restores feasibility. A companion paper provides numerical examples and counterexamples illustrating these results.
\end{abstract}


\section{Introduction}\label{sec:intro}
Control barrier functions (CBFs) provide a systematic framework for enforcing safety constraints in dynamical systems by defining a safe set through a barrier function $h(x)$ and imposing conditions on its evolution to ensure forward invariance \cite{AmesCBF19,XiaoBelta22}. A vast literature has developed along this vein, extending CBF methods to a broad range of safety-critical control problems. Safety concerns are often modeled as constraints on key outputs of the dynamic system.

Simultaneous satisfaction of both input and output constraints has received comparatively less attention \cite{ChenJankovicSantilloAmes21,AgrawalPanagou21,ChoiLeeSreenathTomlinHerbert21,XiaoBeltaCassandras22,IsalyGhanbarpourSanfeliceDixon24,BreedenPanagou23,MousaviTanAmes25,MousaviAmes26,
lavretsky2025servo,FisherAnnaswamy25}. The authors of \cite{AgrawalPanagou21,BreedenPanagou23} certify forward invariance of an inner subset of the safe set, obtained by bounding the worst-case effect of the input limits, while the authors of \cite{ChenJankovicSantilloAmes21} obtain such a subset from a backup controller, and the authors of \cite{IsalyGhanbarpourSanfeliceDixon24} verify the feasibility of multiple CBFs by a sum-of-squares search. All of these methods, however, provide only sufficient conditions and can be computationally demanding.
In \cite{ChoiLeeSreenathTomlinHerbert21} the CBF is constructed from a Hamilton--Jacobi reachability value function, which recovers the maximal safe set under input limits but suffers from the curse of dimensionality.
The authors of \cite{XiaoBeltaCassandras22} introduce an additional CBF that steers the state away from conflicts between the output and input constraints,
but it is unclear whether this additional CBF can itself introduce barriers to constraint satisfaction.
In \cite{MousaviTanAmes25,MousaviAmes26}
it is shown that the CBF constraints admit a feasible input at each state of a set, but that set is not shown to be forward invariant. The paper \cite{FisherAnnaswamy25} handles both constraints, but the class of dynamic systems is limited to a chain of integrators. The authors of \cite{lavretsky2025servo} consider a general LTI system, include integral action, consider both input and output constraints, and propose a specific augmentation-based control architecture for stability and constraint satisfaction. The feasibility of the latter, however, is not established. Distinct from these papers, input and output constraints have been addressed in a combined manner using a Reference Governor approach which consists of modifying the reference input into the system so as to have the system variables belong to an admissible set where all constraints are met \cite{GaroneDiCairanoKolmanovsky17}; determining this admissible set may be computationally burdensome.

In this paper, we derive necessary and sufficient conditions for the feasibility of simultaneously satisfying input and output (I-O) constraints while tracking a command signal using a CBF-based governor for a general class of multi-input LTI systems with state feedback and integral action. Under these conditions, it is shown that the CBF Governor generates a modified command signal that satisfies I-O constraints with all closed-loop solutions remaining bounded in a given set that is forward invariant. The main contribution of the paper is the introduction of a new CBF-based governor with (1) exact characterization of necessary and sufficient conditions that ensure feasibility of I-O constraints with closed-loop boundedness for command tracking, (2) a necessary condition for forward invariance, (3) a design process that seeks feasibility for a given set of input constraints and relaxation thereof only when necessary, and (4) modularity, preserving a closed-loop design intact and modifying only the command.

The problem is posed in \cref{sec:setting}. The I-O CBF Governor is described in \cref{sec:governor}. Necessary and sufficient conditions for the I-O constraints to be met with closed-loop boundedness are established in \cref{sec:farkas}. A design procedure for ensuring that the conditions are always met is described in \cref{sec:design}. We refer the reader to the companion paper \cite{YoonPartII} which provides a number of examples and simulation studies that support the main results of this paper. All notation is listed in Table~\ref{tab:notation} in order of first appearance.

\begin{table*}[t]
\caption{Notation.}
\label{tab:notation}
\centering\footnotesize\setlength{\tabcolsep}{1.6pt}
\begin{tabular}{@{}llll@{}}
\hline
Symbol & Meaning & Symbol & Meaning\\
\hline
\rowcolor{keyrow}$\ker Q=\Oh$ & \textbf{finiteness condition} (\S\ref{sec:finite}) &
$\min_k\nu_k(\Sset)\ge0$ & \textbf{SIOCF condition} on $\Sset$ \eqref{eq:siocfcond}\\
$A_p,B_p,\Clim$ & plant matrices &
$x_p\in\R^{n_p}$ & plant state\\
$u,z\in\R^m$ & plant input, output to be constrained &
$y:=\Creg x_p{+}\Dreg u$ & tracked output\\
$e_I\in\R^m$ & integrator state &
$\ycmd,\yg\in\R^m$ & external command, modified command\\
$\Z$, $\U$ & output box, input box &
$\zmin,\zmax,\umin,\umax$ & lower and upper limits of $z$ and $u$\\
$K_P,K_I$ & PI gains, $\Kfull:=[K_P\ \ K_I]$ &
$u=-\Kfull\xb$ & PI control law \eqref{eq:pi}\\
$\xb\in\R^{n_p+m}$ & augmented state $[x_p\T\ e_I\T]\T$ &
$\Ab,\By$ & augmented loop matrices\\
$\Hu:=\Clim B_p$ & high-frequency gain &
$h_1^{\pm}$, $h_u^{\pm}$ & output, input first-order CBFs\\
$\Cz:=[\Clim\ \ 0_{m\times m}]$ & $z=\Cz\xb$ &
$\psi_1^{\pm},\psi_2^{\pm},\psi_u^{\pm}$ & higher-order CBFs\\
$\alpha_1,\alpha_2,\alpha_u>0$ & CBF gains &
$\ag>0$ & $\yg$ bound\\
$F(\xb)$ & feasible set of \eqref{eq:qp} &
$A_\zeta\in\R^{6m\times m}$ & constraint matrix (constant)\\
$b(\xb)=\Bz\xb+b_0$ & constraint offset &
$\Omega$ & operating set \eqref{eq:omega}\\
$\Sset$ & a nonempty closed subset of $\Omega$ (\cref{def:siocf}) &
$\alpha^{1}$ & design parameter $[\,\alpha_1\ \alpha_2\ \alpha_u\ \ag\,]\T$\\
$\Lambda$, $\lambda^1,\dots,\lambda^N$ & Farkas cone of $A_\zeta$, its extreme rays &
$\nu_k(\Sset)$ & ray values over $\Sset$\\
$Q:=[\Cz\T\ (\Cz\Ab)\T\ \Kfull\T]\T$ & coefficient matrix of $\Omega$ \eqref{eq:omegamat} &
$\Oh$ & unobservable subspace of $(\Ab,\Cbar)$\\
$\Cbar:=[\Cz\T\ \Kfull\T]\T$ & $\Cbar\xb=[z\T\ {-u\T}]\T$ (\S\ref{sec:finite}) &
$V\in\R^{n_p\times d}$ & $\operatorname{span}V=\ker[\Clim\T\ (\Clim A_p)\T]\T$, $\operatorname{rank}V=d$\\
$\etat$, $\eta$ & internal, deviation coordinates &
$E$, $\Emin$ & an RPI set, the minimal one \eqref{eq:emin}\\
$A_\eta$, $\Bd$ & $\eta$ dynamics, $\dot\eta=A_\eta\eta+\Bd\dot z$ \eqref{eq:etadyn} &
$\W$ & set containing $\dot z$ on $\Omega$ \eqref{eq:capset}\\
$w_z$ & width of $\Z$: $\zmax-\zmin$ &
$\SE$ & operating set induced by $E$ (\cref{thm:mainSE})\\
$\alpha^{2}$ & design parameter $[\,\alpha_1\ \alpha_2\ \alpha_u\ \ag\ \varrho\,]\T$ &
$\varrho\ge1$ & relaxation gain in the input constraint\\
$\ag^{\star},\varrho^{\star}$ & thresholds of \cref{thm:fix4} \eqref{eq:thresholds} &
$\sigma_k,\mu_k$ & coefficients of $\ag$ and $\alpha_u\varrho$ in $\lambda^{k\top}b$ \eqref{eq:sigmu}\\
$\varrho_1$ & $\varrho_1:=1$ (\S\ref{sec:fix12}) &
$X_0$, $\alpha_1^{\min}$ & set of initial values, lower bound on $\alpha_1$ \eqref{eq:a1min}\\
$\ThH(\alpha_1^{\min},\varrho)$ & bound on $\ThB$ for $\alpha_1\ge\alpha_1^{\min}$ &
$\ThB(\alpha_1,\varrho)$ & supremum of $\min_k\nu_k$ over the second class \eqref{eq:ThB}\\
$\bar\alpha_1$, $\alpha_1^{\star}$ & upper limit and choice of $\alpha_1$ \eqref{eq:a1bar}, \eqref{eq:a1star} &
$\underline\alpha_2$, $\underline\alpha_g$ & lower bounds of Step~2 \eqref{eq:a2min}, \eqref{eq:agmin}\\
\hline
\end{tabular}
\end{table*}

\section{Problem Setting}\label{sec:setting}

We consider a class of LTI plants with integral action:
\begin{align}
\dot x_p&=A_px_p+B_pu,\notag\\
y&=\Creg x_p+\Dreg u,\label{eq:plant}\\
\dot e_I&=y-\yg ,\notag
\end{align}
where $x_p\in\R^{n_p}$ is the plant state, $u\in\R^m$ is the control input, $y\in\R^m$ is the output required to track a desired signal $\ycmd\in\R^m$, and $e_I\in\R^m$ is the integrator state.
$\yg$ in \eqref{eq:plant} represents a modified command to be chosen so that constraints on the input $u$ and relevant outputs of the system $z$ are satisfied at all times.
The latter is specified as
\begin{subequations}\label{eq:req}
\begin{alignat}{2}
u(t)&\in\U:=\{u\in\R^m:\umin\le u\le\umax\}&\ \ &\forall t\ge t_0,\label{eq:requ}\\
z(t)&\in\Z:=\{z\in\R^m:\zmin\le z\le\zmax\}&\ \ &\forall t\ge t_0,\label{eq:reqz}
\end{alignat}
\end{subequations}
where
\begin{equation}\label{eq:zout}
z=\Clim x_p ,
\end{equation}
with $\Clim\in\R^{m\times n_p}$.
In \eqref{eq:req}, $\zmin,\zmax,\umin,\umax\in\R^m$, and the inequalities are imposed component-wise, where $\umin<0<\umax$ and $\zmin<0<\zmax$.
The overall objective of this paper is that all states remain bounded, $y$ tracks $\ycmd$ as closely as possible, and that constraints on $u$ and $z$ be satisfied at all times.
In order to achieve the first objective, we propose a PI controller of the form
\begin{equation}\label{eq:pi}
u=-K_Px_p-K_Ie_I ,
\end{equation}
where $K_P\in\R^{m\times n_p}$ and $K_I\in\R^{m\times m}$.
In order to achieve the latter two objectives, we propose a CBF-based governor, which is the focus of this paper.

We rewrite the closed-loop system using \eqref{eq:plant} and \eqref{eq:pi} as
\begin{gather}
\dot\xb=\Ab\xb+\By\,\yg ,\qquad
\xb:=\begin{bmatrix}x_p\\ e_I\end{bmatrix}\in\R^{n_p+m},\label{eq:loop}\\
\Ab:=\begin{bmatrix}A_p{-}B_pK_P & {-}B_pK_I\\ \Creg{-}\Dreg K_P & {-}\Dreg K_I\end{bmatrix},
\quad
\By:=\begin{bmatrix}0_{n_p\times m}\\ -I_m\end{bmatrix}.\notag
\end{gather}
We formally state the focus of the paper in \cref{prob:setting}.

\begin{problem}\label{prob:setting}
Given the closed-loop system \eqref{eq:loop} and the constraint requirement \eqref{eq:req}, find a modified command $\yg$ such that, for all bounded piecewise continuous $\ycmd$, every solution of \eqref{eq:loop} remains bounded and \eqref{eq:req} is satisfied.
\end{problem}

The following assumptions are made regarding the underlying LTI system, where $\Hu:=\Clim B_p$.

\begin{assumption}\label{ass:standing}\leavevmode
\begin{enumerate}[label=(\textnormal{A\arabic*}),leftmargin=2.6em,topsep=2pt]
\item\label{ass:A1} $\Hu$ is nonsingular;
\item\label{ass:A2} the invariant zeros of $(A_p,\!B_p,\!\Clim,\!0_{m\times m})$ lie in $\mathbb C^{-}$;
\item\label{ass:A3} $\Ab$ is Hurwitz.
\end{enumerate}
\end{assumption}

Assumptions \textnormal{\ref{ass:A1}} and \textnormal{\ref{ass:A2}} pertain only to the output $z$ to be constrained and impose no restriction on the tracked output $y$.
Assumption \textnormal{\ref{ass:A3}} can be easily satisfied by the choice of $(K_P,K_I)$ whenever \eqref{eq:plant} is stabilizable.

\section{I-O CBF Governor}\label{sec:governor}

The I-O CBF Governor we propose modifies the desired command $\ycmd$ as $\yg$ so as to ensure that the I-O constraints are met in \eqref{eq:req} and so that $\yg$ remains as close as possible to $\ycmd$ (see \autoref{fig:blockdiagram}).
To achieve these objectives, we begin with two CBF pairs $h_1^{\pm}$, $h_u^{\pm}$ as
\begin{align*}
h_1^{+}(\xb)&:=\zmax-\Cz\xb, & h_1^{-}(\xb)&:=\Cz\xb-\zmin,\\
h_u^{+}(\xb)&:=\umax+\Kfull\xb, & h_u^{-}(\xb)&:=-\Kfull\xb-\umin ,
\end{align*}
where $\Cz:=[\Clim\ \ 0_{m\times m}]$ and $\Kfull:=[K_P\ \ K_I]$.
Due to the presence of two integrators in \eqref{eq:plant}, two high-order CBFs (HOCBFs) are constructed pertaining to the state, and due to the integral controller in \eqref{eq:pi}, a first-order CBF is constructed pertaining to the input.
These are defined as
\begin{align*}
\psi_1^{\pm}&:=\dot h_1^{\pm}+\alpha_1h_1^{\pm}, & \psi_2^{\pm}&:=\dot\psi_1^{\pm}+\alpha_2\psi_1^{\pm},\\
\psi_u^{\pm}&:=\dot h_u^{\pm}+\alpha_uh_u^{\pm},
\end{align*}
where $\alpha_1,\alpha_2,\alpha_u>0$.
With these definitions, the I-O CBF Governor will be constructed using the inequalities $\psi_2^{\pm}\ge0$ and $\psi_u^{\pm}\ge0$ for all $t\ge t_0$, much like the approaches adopted in \cite{XiaoBelta22,FisherAnnaswamy25}, but with an additional constraint on $\yg$ to provide a feasible solution.
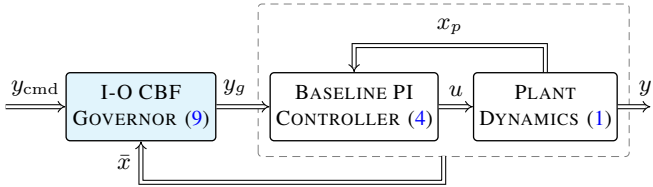
\begin{figure}[t]
\centering
\begin{tikzpicture}[
  font=\fontsize{8}{9.6}\selectfont,
  block/.style={draw, rounded corners=1.5pt, line width=0.5pt, minimum height=0.9cm, align=center, inner sep=2pt},
  vec/.style={double, double distance=1.1pt, line width=0.4pt, -{Implies}},
  lbl/.style={font=\fontsize{9}{10.8}\selectfont, inner xsep=1.5pt, inner ysep=4pt}]
\node[block, minimum width=1.98cm, fill=cyan!10] (gov) at (1.790,0) {\textsc{I-O CBF}\\\textsc{Governor}~\eqref{eq:qpproj}};
\node[block, minimum width=2.22cm] (pi) at (4.610,0) {\textsc{Baseline PI}\\\textsc{Controller}~\eqref{eq:pi}};
\node[block, minimum width=1.9cm] (plant) at (7.140,0) {\textsc{Plant}\\\textsc{Dynamics}~\eqref{eq:plant}};
\draw[gray, line width=0.35pt, densely dashed, rounded corners=0.5mm]
  ([xshift=-1.5mm, yshift=-2mm]pi.south west) rectangle ([xshift=1.5mm, yshift=9mm]plant.north east);
\draw[vec] (0,0) -- node[lbl, above] {$y_{\mathrm{cmd}}$} (gov.west);
\draw[vec] (gov.east) -- node[lbl, above, pos=0.35] {$y_g$} (pi.west);
\draw[vec] (pi.east) -- node[lbl, above] {$u$} (plant.west);
\draw[vec] (plant.east) -- node[lbl, above, pos=0.84] {$y$} (8.53,0);
\draw[vec] (plant.north) -- ++(0,0.35) -| node[lbl, above, pos=0.25] {$x_p$} (pi.north);
\draw[vec] (5.795,-0.65) -- (5.795,-1.0) -| node[lbl, left=2pt, pos=0.75] {$\bar x$} (gov.south);
\end{tikzpicture}
\caption{I-O CBF Governor architecture. The I-O CBF Governor
modifies the desired command $y_{\mathrm{cmd}}$ to produce $y_g$. The augmented
closed-loop state $\bar{x}$ is fed back to the governor.}
\label{fig:blockdiagram}
\end{figure}
Notation-wise, $\psi^{\pm}\ge0$ represents a combined set of inequalities $\psi^{+}\ge0$ and $\psi^{-}\ge0$.
The complete I-O CBF Governor is proposed as
\begin{subequations}\label{eq:qp}
\begin{align}
&\yg(\xb,\ycmd):=\argmin_{\zeta\in\R^m}\ \lVert\zeta-\ycmd\rVert^{2}\label{eq:qp-obj}\\
&\text{s.t.}\ \ \psi_2^{\pm}(\xb,\zeta)\ge0,\ \psi_u^{\pm}(\xb,\zeta)\ge0,\ \lVert\zeta\rVert_\infty\le\ag ,\label{eq:qp-con}
\end{align}
\end{subequations}
where $\ag>0$, $\lVert\cdot\rVert$ denotes the Euclidean norm, and matrix norms are induced.
We rewrite the inequalities \eqref{eq:qp-con} as $\zeta\in F(\xb)$, where the pointwise feasible set \cite{BoydV04} is
\begin{equation}\label{eq:qpF}
F(\xb):=\big\{\zeta\in\R^m:\ A_\zeta\,\zeta\le b(\xb)\big\} ,
\end{equation}
with $b(\xb):=\Bz\xb+b_0$ and the constant matrices
\begin{gather}
A_\zeta:=\begin{bmatrix}\Hu K_I\\ -\Hu K_I\\ K_I\\ -K_I\\
I_m\\ -I_m\end{bmatrix},
\qquad
b_0:=\begin{bmatrix}
\alpha_1\alpha_2\,\zmax\\ -\alpha_1\alpha_2\,\zmin\\
\alpha_u\,\umax\\ -\alpha_u\,\umin\\
\ag\ones_m\\ \ag\ones_m
\end{bmatrix},\notag\displaybreak[0]\\[6pt]
\Bz:=\begin{bmatrix}
-\big(\Cz\Ab^{2}+(\alpha_1{+}\alpha_2)\Cz\Ab+\alpha_1\alpha_2\Cz\big)\\
+\big(\Cz\Ab^{2}+(\alpha_1{+}\alpha_2)\Cz\Ab+\alpha_1\alpha_2\Cz\big)\\
\Kfull(\Ab+\alpha_uI)\\ -\Kfull(\Ab+\alpha_uI)\\ 0\\ 0
\end{bmatrix}.\label{eq:qpdata}
\end{gather}
We denote $F(\xb(t))\neq\varnothing$, $\forall t\ge t_0$, to be the \emph{feasibility condition}.
Using \eqref{eq:qpF}, we rewrite the I-O CBF Governor in \eqref{eq:qp-obj}--\eqref{eq:qp-con} as
\begin{equation}\label{eq:qpproj}
\yg(\xb,\ycmd)=\argmin_{\zeta\in F(\xb)}\ \lVert\zeta-\ycmd\rVert^{2}.
\end{equation}
It is clear that if the feasibility condition is satisfied, then all I-O constraints \eqref{eq:req} are satisfied.
We note that \eqref{eq:qpproj} is a strictly convex quadratic program with $F(\xb)$ defining a compact convex polytope.
Feasibility of \eqref{eq:qpproj} also implies that the solution $\yg$ is continuous and piecewise affine, hence Lipschitz in $\xb$ and $\ycmd$, and unique \cite[Thm.~4]{Bemporad02}.

\subsection{Restatement of Problem~\ref{prob:setting}}\label{sec:goals}

Before we restate \cref{prob:setting}, we introduce two sets, $\Omega$ and $\Sset$, and a definition.
The set $\Omega$ is defined as
\begin{equation}\label{eq:omega}
\Omega:=\{\xb\in\R^{n_p+m}:\ h_1^{\pm}(\xb)\ge0,\ \psi_1^{\pm}(\xb)\ge0,\ h_u^{\pm}(\xb)\ge0\},
\end{equation}
which will be used to establish forward invariance of the closed-loop system \eqref{eq:loop}, i.e.\ if $\xb(t_0)\in\Omega$, then $\xb(t)\in\Omega$ $\forall t\ge t_0$.
The design parameter $\alpha^{1}\in\R_{++}^{4}$ of the I-O CBF Governor, where $\R_{++}$ denotes the positive reals, is given by
\begin{equation}\label{eq:alpha1}
\alpha^{1}:=[\,\alpha_1\ \ \alpha_2\ \ \alpha_u\ \ \ag\,]\T .
\end{equation}
We now introduce a definition, SIOCF, that formalizes the property of feasibility:

\begin{definition}\label{def:siocf}
Let $\Sset\subseteq\Omega$ be nonempty and closed.
The I-O CBF Governor \eqref{eq:qpproj} \textbf{\emph{satisfies simultaneous input and output constraint feasibility (SIOCF)}} on $\Sset$ if, for all bounded piecewise continuous $\ycmd$ and $\forall\xb(t_0)\in\Sset$, the solutions of \eqref{eq:loop} satisfy the following:
\begin{enumerate}[label=(P\arabic*),leftmargin=2.6em]
\item\label{def:P1} existence and uniqueness of $\yg$ on $[t_0,\infty)$;
\item\label{def:P2} $\xb(t)\in\Sset$ $\ \forall t\ge t_0$;
\item\label{def:P3} $z(t)\in\Z$ and $u(t)\in\U$ $\ \forall t\ge t_0$;
\item\label{def:P4} $\lVert\yg(t)\rVert_\infty\le\ag$ $\ \forall t\ge t_0$.
\end{enumerate}
\end{definition}

We now restate \cref{prob:setting}:

\begin{problem}\label{prob:cert}
Given the closed-loop system \eqref{eq:loop}, the constraint requirement \eqref{eq:req}, and the I-O CBF Governor \eqref{eq:qpproj}, find a nonempty $\Sset\subseteq\Omega$ such that the I-O CBF Governor satisfies SIOCF on $\Sset$ for a given $\alpha^{1}$.
\end{problem}

\section{Solution to Problem~\ref{prob:cert}}\label{sec:farkas}

\subsection{Simplification of the Feasibility Condition}\label{sec:farkasred}

The first step in providing a solution to \cref{prob:cert} is a simplification of the feasibility condition.
The following definition is useful.
Define a polyhedral cone $\Lambda$ as
\begin{equation}\label{eq:cone}
\Lambda:=\{\lambda\in\R^{6m}:\ A_\zeta\T\lambda=0,\ \lambda\ge0\} .
\end{equation}
By the Minkowski--Weyl theorem \cite{Schrijver86}, $\Lambda$ can be rewritten using a finite set of \emph{extreme rays} $\lambda^k$, $k=1,\dots,N$, as
\begin{equation}\label{eq:rays}
\Lambda=\Big\{\textstyle\sum_{k=1}^{N}c_k\lambda^k:\ c_k\ge0,\ \lVert\lambda^k\rVert_2=1\Big\} .
\end{equation}
The set $\{\lambda^1,\dots,\lambda^{N}\}$ is unique.
Using these extreme rays, we provide a simpler necessary and sufficient condition of feasibility of \eqref{eq:qpproj} in \cref{prop:cert}.

\begin{proposition}\label{prop:cert}
For nonempty $\Sset\subseteq\R^{\nb}$ let
\[
\nu_k(\Sset):=\inf_{\xb\in\Sset}\lambda^{k\top}b(\xb),\qquad
k=1,\dots,N .
\]
Then $F(\xb)\neq\varnothing$ $\forall\xb\in\Sset$ if and only if
\begin{equation}\label{eq:siocfcond}
\min_k\nu_k(\Sset)\ge0 .
\end{equation}
If in addition $\Hu$ satisfies \textnormal{\ref{ass:A1}} and $\Ab$ in \eqref{eq:loop} satisfies \textnormal{\ref{ass:A3}}, then
\begin{equation}\label{eq:raycount}
9m\ \le\ N\ \le\ 3m+2\binom{3m}{m+1} .
\end{equation}
\end{proposition}

We shall denote \eqref{eq:siocfcond} as the \emph{SIOCF condition} on $\Sset$.
Proofs of the results are given in the Appendix.
\subsection{Necessary Condition for SIOCF on \texorpdfstring{$\Omega$}{Ω}}\label{sec:finite}

The SIOCF condition~\eqref{eq:siocfcond} on $\Omega$ requires every $\nu_k(\Omega)$ to be finite. We therefore introduce a condition equivalent to this finiteness, which is thus necessary for SIOCF on $\Omega$.
For this purpose, we substitute \eqref{eq:loop} into $\psi_1^{\pm}$ to obtain
\begin{equation}\label{eq:psi1}
\psi_1^{\pm}(\xb)=\alpha_1h_1^{\pm}(\xb)\mp\Cz\Ab\xb ,
\end{equation}
and rewrite $\Omega$ as
\begin{equation}\label{eq:omegamat}
\begin{aligned}
\Omega=\{\xb\in\R^{\nb}:\ &T_QQ\,\xb+c_0^{-}\ge0,\\
&-T_QQ\,\xb+c_0^{+}\ge0\} ,
\end{aligned}
\end{equation}
where
\begin{gather*}
Q:=\begin{bmatrix}\Cz\\ \Cz\Ab\\ \Kfull\end{bmatrix},\qquad
T_Q:=\begin{bmatrix}I_m&0&0\\ \alpha_1I_m&I_m&0\\ 0&0&-I_m\end{bmatrix},\\
c_0^{-}:=-[\zmin\T\,\alpha_1\zmin\T\,\umin\T]\T,\ c_0^{+}:=[\zmax\T\,\alpha_1\zmax\T\,\umax\T]\T .
\end{gather*}
It is clear from \eqref{eq:omegamat} that $\Omega$ is unbounded whenever $\dim\ker Q\ge1$, since $T_Q$ is nonsingular and therefore $\ker(T_QQ)=\ker Q$.

We also introduce the set $\Oh$ defined as
\begin{equation}\label{eq:unobs}
\Oh:=\bigcap_{k\ge0}\ker(\Cbar\Ab^{k}) ,
\end{equation}
where $\Cbar:=[\Cz\T\ \ \Kfull\T]\T$.
Since $\Cbar\xb=[z\T\ \ {-u\T}]\T$, $\Oh$ corresponds to the unobservable subspace of $(\Ab,\Cbar)$.
An important observation is that any $v\in\ker Q$ cannot be sensed by the barrier functions in \eqref{eq:omega}. Therefore, a condition that $\ker Q=\Oh$ allows such states $v$ to be unobservable in $z$ and $u$.
We shall denote this to be the \emph{finiteness condition}.
\Cref{prop:escape} formalizes this property.

\begin{proposition}\label{prop:escape}
Assume $\ker Q\supsetneq\Oh$.
$\forall M>0$ and $\forall\bar\tau>0$, $\exists\,(t^\ast,v,r^\ast)\in(0,\bar\tau]\times\ker Q\times(0,\infty)$ such that $\forall r>r^\ast$ and for all piecewise continuous $\zeta$ with $\lVert\zeta(t)\rVert_\infty\le M$, the solution of $\dot\xb=\Ab\xb+\By\,\zeta(t)$ satisfies
\[
\xb(t_0)=r\,v\in\Omega,\qquad
\xb(t_0+t^\ast)\notin\Omega .
\]
\end{proposition}

\Cref{prop:escape} underscores the need for the finiteness condition: if it is violated, trajectories can leave $\Omega$ for every bounded $\ycmd$.
In \cref{thm:fin}, we show that requiring every $\nu_k(\Omega)$ to be finite is equivalent to the finiteness condition.

\begin{proposition}\label{thm:fin}
$\forall\alpha^{1}\in\R_{++}^{4}$,
\[
\nu_k\big(\Omega\big)>-\infty\ \ \forall k=1,\dots,N
\;\;\Longleftrightarrow\;\;
\ker Q=\Oh .
\]
\end{proposition}

\Cref{prop:kerQO} simplifies the finiteness condition to a single zero test of a $2m\times d$ matrix.

\begin{proposition}\label{prop:kerQO}
Define $V\in\R^{n_p\times d}$ such that
\[
\operatorname{span}V=\ker\begin{bmatrix}\Clim\\ \Clim A_p\end{bmatrix},\qquad\operatorname{rank}V=d .
\]
Then, if $\Ab$ satisfies \textnormal{\ref{ass:A3}},
\[
\ker Q=\Oh \iff
\begin{bmatrix}\Clim A_p^{2}\\ K_PA_p+K_I\Creg\end{bmatrix}V=0 .
\]
\end{proposition}

An important point to note is that $\operatorname{rank}Q\le3m$.
Since this implies that $\dim\ker Q\ge n_p-2m$, it follows that the finiteness condition is violated for any plant of order $n_p>2m$ whose pair $(\Ab,\Cbar)$ is observable.
That is, for all such high-order plants one can never guarantee that $\Omega$ remains a compact set.
In \S\ref{sec:rpi} we show how, when $d>0$,\footnote{When $d=0$, e.g.\ when $n_p=m$, \eqref{eq:omegamat} implies that $\Omega$ is compact.} one can construct a compact set $\Sset\subseteq\Omega$ that can be made forward invariant using the I-O CBF Governor.

\subsection{Necessary and Sufficient Conditions for SIOCF}\label{sec:rpi}

In order to find a compact subset $\Sset$ in $\Omega$ that is forward invariant, we transform the state $\xb$ as $\xi=\Tz\xb$, where $\xi=[\,z\T\ \ \dot z\T\ \ \eta\T\,]\T$.
We will determine the compact set $\Sset$ by determining a region $E$ where the unconstrained state $\eta$ will remain forward invariant.
In \cref{lem:normalform}, we describe the construction of $\eta$.
Let $W_p\in\R^{(n_p-m)\times n_p}$ have rows spanning $\ker(B_p\T)$, and define $\etat:=W_px_p\in\R^{n_p-m}$.

\begin{lemma}\label{lem:normalform}
Under \textnormal{\ref{ass:A1}}--\textnormal{\ref{ass:A3}}, the following hold.
\begin{enumerate}[label=(\alph*),leftmargin=2.2em]
\item\label{lem:nf-a} $T:=[\Clim\T\ \ W_p\T]\T$ is nonsingular;
\item\label{lem:nf-b} $\dot\etat=A_\eta\,\etat+B_\eta\,z$,\\
where $A_\eta:=W_pA_pR_2$ and $B_\eta:=W_pA_pR_1$ with the partition $T^{-1}=:[R_1\ \ R_2]$, $R_1\in\R^{n_p\times m}$ and $R_2\in\R^{n_p\times(n_p-m)}$;
\item\label{lem:nf-c} $A_\eta$ is Hurwitz;
\item\label{lem:nf-d} \mbox{}\hfill$\dot\eta=A_\eta\,\eta+\Bd\,\dot z ,$\hfill\refstepcounter{equation}\label{eq:etadyn}\textup{(\theequation)}\newline
where $\eta:=\etat+\Bd\,z=\Wt x_p$ with $\Bd:=A_\eta^{-1}B_\eta$ and $\Wt:=W_p+\Bd\Clim$;
\item\label{lem:nf-e} $\xi:=\Tz\,\xb$, where $\Tz:=\big[\Cz\T\ \ (\Cz\Ab)\T\ \ [\,\Wt\ \ 0\,]\T\big]\T$ is nonsingular.
\end{enumerate}
\end{lemma}
It should be noted that \cref{lem:normalform}\ref{lem:nf-c} follows from \textnormal{\ref{ass:A2}}, since the eigenvalues of $A_\eta$ are the invariant zeros of $(A_p,B_p,\Clim)$ \cite{Isidori95}.
\Cref{lem:normalform} in turn is utilized to prove \cref{thm:mainSE}.

The input $\dot z$ in \eqref{eq:etadyn} satisfies
\begin{equation}\label{eq:capset}
\dot z(\xb)\in\W\qquad\forall\xb\in\Omega ,
\end{equation}
where $\W:=\{w\in\R^m:\ -\alpha_1w_z\le w\le\alpha_1w_z\}$ and $w_z:=\zmax-\zmin$.
This follows from the inequalities in $\Omega$ related to $h_1^{\pm}$ and $\psi_1^{\pm}$.
The property \eqref{eq:capset} in turn is used to establish positive invariance of $\eta$ in \eqref{eq:etadyn}, and therefore in solving \cref{prob:cert}.
The following definition is needed:

\begin{definition}[Robust positive invariance
{\normalfont\cite{BlanchiniMiani15}}]\label{def:rpi} A compact convex $E\subset\R^{n_p-m}$ with $0\in E$ is \emph{robustly positively invariant (RPI)} if every solution of
\[
\dot{\eta}=A_\eta\eta+\Bd\,w,\qquad w(t)\in\W ,
\]
is such that $\eta(t_0)\in E\Rightarrow\eta(t)\in E$ $\ \forall t\ge t_0$.
\end{definition}

We now provide a complete solution to \cref{prob:cert} in \cref{thm:mainSE}, which corresponds to the first main contribution of the paper.

\begin{theorem}\label{thm:mainSE}
Let $E$ be any RPI set of \eqref{eq:etadyn} with $\eta(t_0)\in E$, and define the compact set
\begin{equation}\label{eq:SEset}
\SE:=\Omega\cap \{\xb:\ \eta(\xb)\in E\} .
\end{equation}
Then the following hold:
\begin{enumerate}[label=(\alph*),leftmargin=2.2em]
\item\label{thm:SE-a} Let the finiteness condition hold. Then, the I-O CBF Governor satisfies SIOCF on $\Omega$ if and only if $\min_{k}\nu_k(\Omega)\ge0$;
\item\label{thm:SE-b} the I-O CBF Governor satisfies SIOCF on $\SE$ if and only if $\min_{k}\nu_k(\SE)\ge0$.
\end{enumerate}
\end{theorem}

\begin{proof}[Proof of (a)]
($\Rightarrow$) From \cref{def:siocf}, \textnormal{\ref{def:P1}} requires that \eqref{eq:qpproj} have a solution for every $\xb(t_0)\in\Omega$, that is, $F(\xb)\neq\varnothing$ $\forall\xb\in\Omega$, which implies, by \cref{prop:cert}, that the SIOCF condition is satisfied on $\Omega$.

($\Leftarrow$) By \cref{prop:cert}, the SIOCF condition on $\Omega$ implies that $F(\xb)\neq\varnothing$ $\forall\xb\in\Omega$.
By \cite[Thm.~4]{Bemporad02}, $\yg$ is unique and piecewise affine, hence Lipschitz, in $(\xb,\ycmd)$, and the closed-loop solutions of \eqref{eq:loop} are unique for every $\xb(t_0)\in\Omega$ on a maximal interval $[t_0,\tau)$ \cite{Hale80}.
Since $\yg(t)\in F(\xb(t))$, all constraints in \eqref{eq:qp-con} are satisfied on $[t_0,\tau)$, including $\lVert\yg(t)\rVert_\infty\le\ag$, which is \textnormal{\ref{def:P4}}.
The comparison lemma \cite[Lemma~3.4]{Khalil02}, applied successively to $\dot\psi_1^{\pm}\ge-\alpha_2\psi_1^{\pm}$, $\dot h_1^{\pm}\ge-\alpha_1h_1^{\pm}$ and $\dot h_u^{\pm}\ge-\alpha_uh_u^{\pm}$, yields $\psi_1^{\pm},h_1^{\pm},h_u^{\pm}\ge0$ on $[t_0,\tau)$, since $\xb(t_0)\in\Omega$.
That is, $\xb(t)\in\Omega$, $z(t)\in\Z$ and $u(t)\in\U$, which are \textnormal{\ref{def:P2}} and \textnormal{\ref{def:P3}}.
Since $\Ab$ is Hurwitz by \textnormal{\ref{ass:A3}} and $\lVert\yg\rVert_\infty\le\ag$, $\xb(t)$ is bounded on $[t_0,\tau)$.
This allows $[t_0,\tau)$ to be extended to $[t_0,\infty)$ \cite{Hale80}, proving \textnormal{\ref{def:P1}}.

\emph{Proof of (b).}
On $\Omega$, $h_1^{\pm}\ge0$ and $\psi_1^{\pm}\ge0$ imply $z\in\Z$ and $\dot z\in\W$ by \eqref{eq:capset}.
For $\xb\in\SE$, in addition, $\eta(\xb)\in E$ with $E$ compact by \cref{def:rpi}; hence $\xi=\Tz\xb$ is bounded, and $\xb$ is bounded since $\Tz$ is nonsingular by \cref{lem:normalform}\ref{lem:nf-e}.
$\SE$ is closed as the intersection of the closed set $\Omega$ with the preimage of the closed set $E$ under $\xb\mapsto\eta(\xb)$, and therefore compact.
Since the solutions remain in $\Omega$ by the proof of (a), $\eta(t)$ remains in $E$ by \eqref{eq:etadyn} and \cref{def:rpi}, which yields $\xb(t)\in\SE$ for all $t\ge t_0$.
Replacing $\Omega$ by $\SE$ in the proof of (a) yields (b).
\end{proof}

\Cref{thm:mainSE}\ref{thm:SE-a} implies that under the finiteness condition, \cref{prob:cert} is solved on $\Sset=\Omega$ if and only if the SIOCF condition is satisfied.
\Cref{thm:mainSE}\ref{thm:SE-b} implies that \cref{prob:cert} is solved on $\Sset=\SE$, independent of the finiteness condition, if and only if the SIOCF condition is satisfied.

It can be shown that $\SE$ always exists, through the construction of a set $\Emin$ given by
\begin{equation}\label{eq:emin}
\Emin:=\operatorname{cl}\Big\{\int_0^{t}e^{A_\eta\tau}\Bd\,w(\tau)\,d\tau:\
t\ge0,\ w(\tau)\in\W\Big\},
\end{equation}
where $\operatorname{cl}(\cdot)$ denotes the closure operator.
From \eqref{eq:emin}, the integral is bounded, since $A_\eta$ is Hurwitz by \cref{lem:normalform}\ref{lem:nf-c} and $\W$ is bounded.
Given $\eta(t_0)\in\Pieta$, using reachability properties, one can show that an RPI set $E$ exists \cite{RakovicKouramas07} and can be constructed as
\begin{equation}\label{eq:Eset}
\begin{gathered}
E:=\operatorname{cl}\Big\{\operatorname{conv}\Big\{e^{A_\eta(t-t_0)}\eta(t_0)
+\int_{t_0}^{t}e^{A_\eta(t-\tau)}\Bd\,w(\tau)\,d\tau:\\
t\ge t_0,\ w(\tau)\in\W\Big\}\Big\},
\end{gathered}
\end{equation}
where $\operatorname{conv}(\cdot)$ denotes the convex hull and $\Pieta:=\{\eta(\xb):\ \xb\in\Omega\}$.
We also note that, since $E\supseteq\Emin$, $E$ can be enlarged by increasing $\alpha_1$.

\begin{remark}\label{rem:modular}
The I-O CBF Governor in \eqref{eq:qpproj} is designed in a modular manner, tasked only with modifying the command $\ycmd$ as $\yg$ so that the trajectories remain within admissible input and output limits but keeping the closed-loop design unchanged. The difference between $\ycmd$ and $\yg$ may be interpreted as an anti-windup component necessitated by the integral action in the closed loop. The overall control architecture is therefore modular in nature, preserves closed-loop stability and simultaneous I-O constraint satisfaction, and represents another contribution of the paper.
\end{remark}

\Cref{thm:mainSE} requires that the SIOCF condition be satisfied.
This in turn reduces to finding the sign of $\min_k\nu_k(\SE)$, which requires solving $N$ independent problems, one per extreme ray, and hence has a minimal computational burden of $O(N)$; each problem is a linear program for a polyhedral approximation of $E$ or a second-order cone program for an ellipsoidal approximation \cite[Ch.~6]{BlanchiniMiani15}, \cite[Secs.~2.1 and~3]{AlthoffFrehseGirard21}.
We also note that \cref{thm:mainSE} assures a solution to \cref{prob:cert} for a given $\alpha^{1}$.
It is quite possible that no such $\alpha^{1}$ exists for which the SIOCF condition is satisfied, as \cref{prop:fix4nec} shows.

\begin{proposition}\label{prop:fix4nec}
Let $\dot x=ax+u$ with $a>0$, $z=x$, $n_p=m=1$, and let $\Z=[-\zmax,\zmax]$, $\U=[-\umax,\umax]$ with $\umax<a\,\zmax$.
Then the following hold:
\begin{enumerate}[label=(\alph*),leftmargin=2.2em]
\item $\forall z(t_0)\in(\umax/a,\ \zmax]$ and for all piecewise continuous $u(t)\in\U$,
\[
z(t)\ \ge\ \tfrac{\umax}{a}+\big(z(t_0)-\tfrac{\umax}{a}\big)e^{a(t-t_0)},
\]
and the output box is violated in finite time;
\item if $\Ab$ in \eqref{eq:loop} satisfies \textnormal{\ref{ass:A3}}, then
\[
\min_k\nu_k(\SE)\le\tfrac{\alpha_1}{\sqrt2}(\umax-a\,\zmax)<0\ \ \ \forall\alpha^{1}\in\R_{++}^{4} .
\]
\end{enumerate}
\end{proposition}
The main reason for $\nu_k$ to become negative and for the output constraint to be violated is the conflict between input and output constraints, since keeping $z$ at $\zmax$ requires $u=-a\zmax$, and therefore a violation of the input constraint.
It is easy to show that such a violation occurs for some plants in \eqref{eq:plant} with any control policy, and for any $\alpha^{1}$.
We fix this conflict in the next section, where we provide a solution by expanding the design space from $\alpha^{1}$ to a vector $\alpha^{2}$.

\section{SIOCF Design Procedure}\label{sec:design}

We define a new design parameter $\alpha^{2}\in\R_{++}^{5}$ as
\begin{equation}\label{eq:alpha2}
\alpha^{2}:=[\,\alpha_1\ \ \alpha_2\ \ \alpha_u\ \ \ag\ \ \varrho\,]\T ,
\end{equation}
where $\varrho$ is a relaxation gain, with $\varrho\ge1$, that modifies the input constraint in \eqref{eq:req} as
\begin{equation*}\tag{2a$'$}\label{eq:requrho}
u(t)\in\varrho\,\U:=\{u\in\R^m:\ \varrho\,\umin\le u\le\varrho\,\umax\}\ \ \forall t\ge t_0 ,
\end{equation*}
where $\varrho\,\umin<0<\varrho\,\umax$.
We now introduce \cref{prob:design}.

\begin{problem}\label{prob:design}
Given the closed-loop system \eqref{eq:loop} and the constraint requirement \eqref{eq:requrho} and \eqref{eq:reqz}, find an $\alpha^{2}$ and a corresponding RPI set $E$ such that the I-O CBF Governor \eqref{eq:qpproj} satisfies SIOCF on $\SE$, with $\U$ in \textnormal{\ref{def:P3}} replaced by $\varrho\,\U$.
\end{problem}

\Cref{prob:design} is solved in the following sections, where $\alpha^{2}$ is grouped into two classes, $[\,\alpha_1\ \ \alpha_2\ \ \alpha_u\,]\T$ and $[\,\ag\ \ \varrho\,]\T$.
We will show in \S\ref{sec:solvable} that for arbitrary positive values of the constants in the first class, constants $\ag^{\star}$ and $\varrho^{\star}$ always exist in the second class.
In \S\ref{sec:fix12}, we
provide an efficient design for choosing the constants in the first class so as to enlarge $\SE$ while keeping $\varrho$ near unity.

\subsection{Finding an \texorpdfstring{$\alpha^{2}$}{α²}}\label{sec:solvable}

Since $\varrho$ enters the rows $h_u^{\pm}$ of \eqref{eq:omega}, the set $\SE$ varies with $\varrho$ and is denoted as $\SE(\varrho)$.
We now show that \cref{prob:design} is always solvable, which corresponds to the second main contribution of the paper.

\begin{theorem}\label{thm:fix4}
$\forall(\alpha_1,\alpha_2,\alpha_u)\in\R_{++}^{3}$, $\exists\,\ag^{\star},\varrho^{\star}$ such that\\
the I-O CBF Governor satisfies SIOCF on $\SE(\varrho)$ $\forall(\ag,\varrho)\in(\ag^{\star},\infty)\times[\varrho^{\star},\infty)$.
\end{theorem}

\begin{proof}
Choose any $\alpha_1,\alpha_2,\alpha_u>0$ and an RPI set $E$ as in \S\ref{sec:farkas}, and define $(v_1;\dots;v_6):=[v_1\T\,\cdots\,v_6\T]\T$ for $v_j\in\R^m$.
By \cref{thm:mainSE}\ref{thm:SE-b}, it suffices to construct $\ag^{\star}$ and $\varrho^{\star}$ such that, for some $\varepsilon\ge0$, $\nu_k(\SE(\varrho))\ge\varepsilon$ for all $k$ whenever $\ag>\ag^{\star}$ and $\varrho\ge\varrho^{\star}$.
To this end, we classify the extreme rays by partitioning $\{1,\dots,N\}$ into three sets $\Imu$, $\Isig$ and $\Izero$ as follows.
With \eqref{eq:rays} partitioned as
\begin{gather*}
\lambda^{k}=(\lambda^k_{z+};\lambda^k_{z-};\lambda^k_{u+};\lambda^k_{u-};\lambda^k_{M+};\lambda^k_{M-}) ,
\end{gather*}
we then define the three sets
\begin{gather*}
\Imu:=\{k:\ \mu_k>0\},\qquad \Isig:=\{k:\ \mu_k=0,\ \sigma_k>0\},\\
\Izero:=\{k:\ \mu_k=\sigma_k=0\},
\end{gather*}
where
\begin{equation}\label{eq:sigmu}
\sigma_k:=\lambda^{k\top}\eM\ge0,\
\mu_k:=\lambda^{k\top}_{u+}\umax-\lambda^{k\top}_{u-}\umin\ge0,
\end{equation}
and $\eM:=(0;0;0;0;\ones_m;\ones_m)$.
Both $\sigma_k$ and $\mu_k$ are nonnegative, since $\lambda^k\ge0$ and $\umax>0>\umin$ componentwise.
Choose
\begin{equation}\label{eq:thresholds}
\ag^{\star}:=\max_{k\in\Isig}\frac{\varepsilon-\beta_k^{\infty}}{\sigma_k},\qquad
\varrho^{\star}:=\max\Big\{1,\ \max_{k\in\Imu}\frac{\varepsilon-\beta_k^{\infty}}{\alpha_u\,\mu_k}\Big\},
\end{equation}
and $\varepsilon\in[0,\nub]$, where $\nub:=\min_i\tfrac{\alpha_1\alpha_2}{\sqrt2}(w_z)_i>0$.
We then define a set $\SE(\infty)$ and a parameter $\beta_k^{\infty}$ as
\begin{gather*}
\SE(\infty):=\{\xb:\ h_1^{\pm}(\xb)\ge0,\ \psi_1^{\pm}(\xb)\ge0,\ \eta(\xb)\in E\},\\
\beta_k^{\infty}:=\inf_{\xb\in\SE(\infty)}\lambda^{k\top}\big(\Bz\xb+b_0^{z}\big),
\end{gather*}
where $b_0^{z}:=(\alpha_1\alpha_2\zmax;-\alpha_1\alpha_2\zmin;0;0;0;0)$.
It can be noted that $\SE(\infty)$ coincides with $\SE(\varrho)$ with the input constraint removed.
By definition, the first two inequalities in $\SE(\infty)$ imply that $z\in\Z$ and $\dot z\in\W$ by \eqref{eq:capset}, and the third implies that $\eta\in E$.
Hence $\xi=\Tz\xb$ is bounded, and $\SE(\infty)$ is compact by \cref{lem:normalform}\ref{lem:nf-e}.
Moreover, $\SE(\varrho)\subseteq\SE(\infty)$ for all $\varrho\ge1$, since $\SE(\varrho)$ is obtained from $\SE(\infty)$ by adding the inequalities $h_u^{\pm}\ge0$.
Therefore $\beta_k^{\infty}$ is finite for all $k$, and the thresholds in \eqref{eq:thresholds} are finite.

We now show that the SIOCF condition holds on $\SE(\varrho)$ by examining $\nu_k$ in each of the three sets.
For this purpose, we decompose $\lambda^{k\top}b(\xb)$ using \eqref{eq:qpdata}.
With $\U$ replaced by $\varrho\,\U$, $b_0=b_0^{z}+(0;0;\alpha_u\varrho\umax;-\alpha_u\varrho\umin;\ag\ones_m;\ag\ones_m)$, so that $\lambda^{k\top}b(\xb)$ decomposes as
\begin{equation}\label{eq:affinerho}
\lambda^{k\top}b(\xb)=\lambda^{k\top}\big(\Bz\xb+b_0^{z}\big)+\alpha_u\varrho\,\mu_k+\sigma_k\ag .
\end{equation}
Taking the infimum of \eqref{eq:affinerho} over $\SE(\varrho)$, and noting that $\alpha_u\varrho\,\mu_k$ and $\sigma_k\ag$ do not depend on $\xb$, we obtain
\[
\nu_k\big(\SE(\varrho)\big)=\beta_k+\alpha_u\varrho\,\mu_k+\sigma_k\ag ,
\]
where $\beta_k:=\inf_{\xb\in\SE(\varrho)}\lambda^{k\top}(\Bz\xb+b_0^{z})\ge\beta_k^{\infty}$, the inequality following from $\SE(\varrho)\subseteq\SE(\infty)$.
Hence
\begin{equation}\label{eq:nufloor}
\nu_k\big(\SE(\varrho)\big)\ \ge\ \nu_k\big(\SE(\infty)\big)\ =\ \beta_k^{\infty}+\alpha_u\varrho\,\mu_k+\sigma_k\ag
\end{equation}
$\forall\varrho\ge1$ and $\forall\ag>0$.
We now show that $\nu_k(\SE(\varrho))\ge\varepsilon$ for every $k$ whenever $\ag>\ag^{\star}$ and $\varrho\ge\varrho^{\star}$, by considering the three sets separately.

\emph{Case (i): $k\in\Imu$.}
By \eqref{eq:thresholds}, $\alpha_u\varrho^{\star}\mu_k\ge\varepsilon-\beta_k^{\infty}$, and \arxref{eq:nufloor} yields, for all $\varrho\ge\varrho^{\star}$ and all $\ag>0$,
\[
\nu_k\big(\SE(\varrho)\big)\ \ge\ \beta_k^{\infty}+\alpha_u\varrho\,\mu_k+\sigma_k\ag\ \ge\ \beta_k^{\infty}+\alpha_u\varrho^{\star}\mu_k\ \ge\ \varepsilon ,
\]
since $\sigma_k\ag\ge0$.

\emph{Case (ii): $k\in\Isig$.}
By \eqref{eq:thresholds}, $\sigma_k\ag^{\star}\ge\varepsilon-\beta_k^{\infty}$, and \arxref{eq:nufloor} yields, for all $\ag>\ag^{\star}$ and all $\varrho\ge1$,
\[
\nu_k\big(\SE(\varrho)\big)\ \ge\ \beta_k^{\infty}+\sigma_k\ag\ \ge\ \beta_k^{\infty}+\sigma_k\ag^{\star}\ \ge\ \varepsilon .
\]
Cases (i) and (ii) thus give
\begin{equation}\label{eq:numusig}
\nu_k\big(\SE(\varrho)\big)\ \ge\ \varepsilon\qquad\forall k\in\Imu\cup\Isig
\end{equation}
whenever $\ag>\ag^{\star}$ and $\varrho\ge\varrho^{\star}$.

\emph{Case (iii): $k\in\Izero$.}
We first show that $\lambda^k=\chi^{z,i}$ for some $i$, where
\begin{equation}\label{eq:chiz}
\chi^{z,i}:=\tfrac{1}{\sqrt2}(e_i;e_i;0;0;0;0)
\end{equation}
and $e_i\in\R^m$ is the $i$-th standard basis vector.
Since $\lambda^k\ge0$, $\sigma_k=\lambda^{k\top}\eM=0$ implies $\lambda^k_{M+}=\lambda^k_{M-}=0$.
Likewise, since $\umax>0>\umin$ componentwise, $\mu_k=0$ implies $\lambda^k_{u+}=\lambda^k_{u-}=0$.
The condition $A_\zeta\T\lambda^k=0$ in \eqref{eq:cone} then reduces, by \eqref{eq:qpdata}, to
\[
(\Hu K_I)\T\big(\lambda^k_{z+}-\lambda^k_{z-}\big)=0 .
\]
Since $\Hu K_I$ is nonsingular by \textnormal{\ref{ass:A1}} and \textnormal{\ref{ass:A3}}, $\lambda^k_{z+}=\lambda^k_{z-}=:w\ge0$, so that
\[
\lambda^k=(w;w;0;0;0;0)=\sqrt2\sum_{i=1}^{m}w_i\,\chi^{z,i} .
\]
Each $\chi^{z,i}$ belongs to $\Lambda$, since $\chi^{z,i}\ge0$ and $A_\zeta\T\chi^{z,i}=\tfrac{1}{\sqrt2}(\Hu K_I)\T(e_i-e_i)=0$.
An extreme ray is not the sum of two elements of $\Lambda$ that are not its multiples; therefore exactly one $w_i$ is nonzero, and $\lVert\lambda^k\rVert_2=1$ yields $\lambda^k=\chi^{z,i}$.
By \eqref{eq:qpdata}, $\chi^{z,i\top}b(\xb)=\tfrac{1}{\sqrt2}\big(b_{z+}(\xb)+b_{z-}(\xb)\big)_i=\tfrac{\alpha_1\alpha_2}{\sqrt2}(w_z)_i$ for every $\xb$, and hence
\begin{equation}\label{eq:nuzero}
\nu_k=\chi^{z,i\top}b(\xb)=\tfrac{\alpha_1\alpha_2}{\sqrt2}(w_z)_i\ge\nub\ge\varepsilon\quad\forall k\in\Izero .
\end{equation}

Since \eqref{eq:numusig} and \eqref{eq:nuzero} imply that the SIOCF condition is satisfied on $\SE(\varrho)$, \cref{thm:mainSE}\ref{thm:SE-b} proves \cref{thm:fix4}.
\end{proof}

\begin{remark}\label{rem:rays}
By \eqref{eq:nufloor}, the design parameters act on the extreme rays through $\mu_k$ and $\sigma_k$: the corresponding $\nu_k$ increases with $\varrho$ when the ray weights an input constraint ($k\in\Imu$), and otherwise with $\ag$ when it weights \mbox{$\lVert\zeta\rVert_\infty\le\ag$} ($k\in\Isig$).
The remaining rays, those in $\Izero$, are the output pairs \eqref{eq:chiz}, whose $\nu_k$ are positive constants.
The input relaxation $\varrho$ is needed because $\Imu$ also contains rays that weight output constraints against an input constraint with \mbox{$\sigma_k=0$}, such as
\[
\chi^{\times,i}:=\frac{(p^i_+;p^i_-;0;e_i;0;0)}{\sqrt{1+\lVert p^i\rVert_2^2}},\qquad p^i:=(\Hu\T)^{-1}e_i .
\]
Where both constraints are active, $\lambda^{k\top}b$ is independent of $(\alpha_2,\alpha_u,\ag)$; this coupling produces counterexamples such as \cref{prop:fix4nec}, and, for all $(\alpha_1,\alpha_2,\alpha_u)$, the relaxation \mbox{$\varrho\ge\varrho^{\star}$} of \cref{thm:fix4} renders these $\nu_k$ nonnegative.
\end{remark}

In \cref{thm:fix4}, we chose $\ag$ and $\varrho$ by design but the remaining parameters arbitrarily.
The latter may be such that either $\varrho^{\star}$ may be unnecessarily large, or $\SE(\varrho)$ may be small, or both.
We reduce this inefficiency in \S\ref{sec:fix12}.

\subsection{Efficient design of an \texorpdfstring{$\alpha^{2}$}{α²}}\label{sec:fix12}

The design algorithm we propose is summarized in \cref{alg:fix}.
The main idea behind this design is to explore three steps, Steps 1 through 3, with $\varrho$ held at $\varrho_1:=1$ in Steps~1 and~2.
In Step~1, we determine $\alpha_1$, which sets $\Omega$ and $\SE(\varrho)$.
In Step~2, we determine $(\alpha_2,\alpha_u,\ag)$, which only alter the quadratic program \eqref{eq:qp} but do not affect $\Omega$ or $\SE(\varrho)$.
If the SIOCF condition~\eqref{eq:siocfcond} cannot be satisfied in Steps~1 and~2, we proceed to Step~3, where we relax the input constraint by increasing $\varrho$, thereby enlarging $\SE(\varrho)$.
We now describe these three steps, starting with an initialization step.

\begin{algorithm}[t]
\caption{$\alpha^{2}$ design procedure, $\varrho_1:=1$}\label{alg:fix}
\noindent\hspace*{-1.7em}\begin{minipage}{\dimexpr\linewidth+1.7em\relax}
\begin{algorithmic}[1]\makeatletter\def\ALC@lno{}\makeatother
\STATE Compute the extreme rays of $\Lambda$ and $\alpha_1^{\min}$ as in \eqref{eq:a1min}.
\STATE Construct an RPI set $E$ and $\SE(\varrho_1)$.
\STATE \emph{(Step 1)} Compute $\ThH(\alpha_1^{\min},\varrho_1)$ as in \eqref{eq:ThH}; if $\ThH\ge0$, compute $\ThB(\alpha_1,\varrho_1)$ as in \eqref{eq:ThB} over $[\alpha_1^{\min},\bar\alpha_1]$.
\IF{$\ThH\ge0$ and $\ThB(\alpha_1,\varrho_1)>0$ for some $\alpha_1\in[\alpha_1^{\min},\bar\alpha_1]$}
\STATE \emph{(Step 2)} Set $\alpha_1=\alpha_1^{\star}$ as defined in \eqref{eq:a1star}, choose $\varepsilon\in(0,\ThB)\cap\R$ and $c>0$.
\STATE Set $\alpha_2\ge\underline\alpha_2$ with $\alpha_u=c\,\alpha_2$, and $\ag\ge\underline\alpha_g(\alpha_2)$ defined in \eqref{eq:a2min} and \eqref{eq:agmin}.
\ELSE
\STATE \emph{(Step 3)} Choose $\varepsilon\in[0,\nub]$, and set $\ag>\ag^{\star}$ and $\varrho\ge\varrho^{\star}$ defined in \eqref{eq:thresholds}.
\ENDIF
\end{algorithmic}
\end{minipage}
\end{algorithm}

\emph{Initialization.}
We choose
\begin{equation}\label{eq:a1min}
\alpha_1^{\min}:=\inf\big\{\alpha_1>0:\ \psi_1^{\pm}(\xb)\ge0\ \ \forall\xb\in X_0\big\}
\end{equation}
and restrict $\alpha_1$ to $[\alpha_1^{\min},\infty)$, where $X_0$ is the set of all initial values $\xb(t_0)$.
Such a choice ensures that $X_0\subseteq\Omega$.
We then construct an RPI set $E$ using reachability properties of \eqref{eq:etadyn} and a polyhedral or an ellipsoidal approximation \cite[Ch.~6]{BlanchiniMiani15}, \cite[Sec.~3]{AlthoffFrehseGirard21}, and the corresponding $\SE(\varrho_1)$.

We set up a few preliminaries and a Proposition before proceeding to the steps.

Let $\hat b$, $\Psi$, $H$ be the blocks of the decomposition
\begin{equation}\label{eq:bsplit}
b(\xb)=\hat b(\xb)+\alpha_2\,\Psi(\xb)+\alpha_u\,H(\xb)+\ag\,\eM ,
\end{equation}
given by
\begin{align*}
\hat b(\xb)&:=\big(\mp\Cz(\Ab+\alpha_1I)\Ab\xb;\ \pm\Kfull\Ab\xb;\ 0_m;\ 0_m\big),\\
\Psi(\xb)&:=\big(\psi_1^{+}(\xb);\ \psi_1^{-}(\xb);\ 0;\ 0;\ 0;\ 0\big),\\
H(\xb)&:=\big(0;\ 0;\ h_u^{+}(\xb);\ h_u^{-}(\xb);\ 0;\ 0\big) .
\end{align*}
Let
\[
Z_k:=\big\{\xb\in\SE(\varrho):\ \lambda^{k\top}\Psi(\xb)=0,\ \lambda^{k\top}H(\xb)=0\big\} ,
\]
and $H_k:=\{\xb\in Z_k:\ \Cz\Ab\xb=0\}$.
Compute
\begin{equation}\label{eq:ThH}
\ThH(\alpha_1^{\min},\varrho):=\min_{k\in\Kz}\ \inf_{\xb\in H_k}\lambda^{k\top}\hat b(\xb)
\end{equation}
with $Z_k$ and $\hat b$ evaluated on $\SE(\varrho)$ with $\alpha_1=\alpha_1^{\min}$, and $\Kz:=\{k:\ \sigma_k=0\}$ with $\sigma_k$ as in \eqref{eq:sigmu}.
Define
\begin{equation}\label{eq:ThB}
\ThB(\alpha_1,\varrho):=\min_{k\in\Kz}\ \inf_{\xb\in Z_k}\lambda^{k\top}\hat b(\xb).
\end{equation}
Also define $\bar\alpha_1$ as
\begin{equation}\label{eq:a1bar}
\bar\alpha_1:=\sup\{\alpha_1\ge\alpha_1^{\min}:\ \alpha_1\notin A\ \text{ or }\ q(\alpha_1)\ge0\}
\end{equation}
where
\begin{equation}\label{eq:qdef}
q(\alpha_1):=q_2\alpha_1^2+q_1\alpha_1+q_0
\end{equation}
with $q_2<0$, and a set $A\subseteq\R_{++}$ such that
\begin{equation}\label{eq:qbound}
\ThB(\alpha_1,\varrho)\le q(\alpha_1)\qquad\forall\alpha_1\in A .
\end{equation}

We then have the following proposition.
\begin{proposition}\label{prop:ceiling}
For $\varrho$, $\alpha_1^{\min}$, and an RPI set $E$,
\begin{flalign*}
\textnormal{\hypertarget{prop6a}{(a)}} && \sup_{\alpha_2,\alpha_u,\ag>0}\min_k\nu_k(\SE(\varrho))=\ThB(\alpha_1,\varrho)\quad\forall\alpha_1>0 ; &&
\end{flalign*}
\begin{flalign*}
\textnormal{\hypertarget{prop6b}{(b)}} && \ThB(\alpha_1,\varrho)\ \le\ \ThH(\alpha_1^{\min},\varrho)\qquad\forall\alpha_1\in[\alpha_1^{\min},\infty) ; &&
\end{flalign*}
\par\noindent\hangindent=2.2em\hangafter=1\makebox[2.2em][l]{\textnormal{\hypertarget{prop6c}{(c)}}}a quadratic function $q(\alpha_1)$ as in \eqref{eq:qdef} and a set $A$ as in \eqref{eq:qbound} exist;
\par\noindent\hangindent=2.2em\hangafter=1\makebox[2.2em][l]{\textnormal{\hypertarget{prop6d}{(d)}}}if $\ThB(\alpha_1,\varrho)>0$, then $\forall c>0$, $\exists\,\underline\alpha_2\in[0,\infty)$ and a nonincreasing function $\underline\alpha_g$ such that the SIOCF condition is satisfied on $\SE(\varrho)$ with $\alpha_u=c\,\alpha_2$ if and only if $\alpha_2\ge\underline\alpha_2$ and $\ag\ge\underline\alpha_g(\alpha_2)$.
\end{proposition}

\Cref{prop:ceiling} guides the choice of $\alpha_1$ as well as $(\alpha_2,\alpha_u,\ag)$. \Cref{prop:ceiling}\hyperlink{prop6a}{(a)} implies that $(\alpha_2,\alpha_u,\ag)$ exist that satisfy the  SIOCF condition  for a given $\alpha_1$ if $\ThB(\alpha_1,\varrho)>0$ and do not exist if $\ThB(\alpha_1,\varrho)<0$. \hyperlink{prop6b}{(b)} implies that such an $\alpha_1\ge\alpha_1^{\min}$ cannot exist when $\ThH<0$, and \hyperlink{prop6c}{(c)} implies that such an $\alpha_1$ cannot exceed $\bar\alpha_1$. \hyperlink{prop6d}{(d)} finally says that for the chosen $\alpha_1$,  thresholds for $(\alpha_2,\alpha_u,\ag)$ can be suitably determined so that the SIOCF condition holds.

We now proceed with the description of the three steps.

\noindent \emph{Step 1.}
If $\ThH<0$, proceed to Step~3.
Otherwise compute $\ThB(\alpha_1,\varrho_1)$ over $\alpha_1\in[\alpha_1^{\min},\bar\alpha_1]$, recomputing $E$ and $\SE(\varrho_1)$.
If $\ThB(\alpha_1,\varrho_1)>0$ for some $\alpha_1$, set $\alpha_1=\alpha_1^{\star}$, where
\begin{equation}\label{eq:a1star}
\alpha_1^{\star}:=\max\{\alpha_1\in[\alpha_1^{\min},\bar\alpha_1]:\ \ThB(\alpha_1,\varrho_1)>0\} ;
\end{equation}
and proceed to Step~2; otherwise, proceed to Step~3.
Since $\Omega$ and $E$ both grow as $\alpha_1$ increases, so does $\SE(\varrho_1)$; hence $\alpha_1^{\star}$ in \eqref{eq:a1star} yields the largest $\SE(\varrho_1)$ among all $\alpha_1$ with $\ThB(\alpha_1,\varrho_1)>0$.
If $\bar\alpha_1=\infty$, the maximum in \eqref{eq:a1star} need not exist, in which case, any $\alpha_1$ with $\ThB(\alpha_1,\varrho_1)>0$ is chosen.

\noindent \emph{Step 2.}
Choose $\varepsilon\in(0,\ThB)\cap\R$ and $c>0$, and set $\alpha_2\ge\underline\alpha_2$ with $\alpha_u=c\,\alpha_2$ and $\ag\ge\underline\alpha_g(\alpha_2)$, with $\underline\alpha_2$ and $\underline\alpha_g$ defined as in \eqref{eq:a2min} and \eqref{eq:agmin}.

\noindent\emph{Step 3.}
With any $\alpha_1\ge\alpha_1^{\min}$ and $(\alpha_2,\alpha_u)$,
followed by the construction of $E$, choose $\varepsilon\in[0,\nub]$, and set $\ag>\ag^{\star}$ and $\varrho\ge\varrho^{\star}$, where the lower bounds are as in \eqref{eq:thresholds}.

If Step~2 is reached, \cref{prop:ceiling}\hyperlink{prop6d}{(d)} yields the SIOCF condition on $\SE(\varrho_1)$; if Step~3 is reached, \cref{thm:fix4} yields it on $\SE(\varrho)$.
For a polyhedral approximation of $E$ and a precomputed set of extreme rays, evaluation of the SIOCF condition requires $N$ independent linear programs.

\section{Conclusions}\label{sec:concl}
This paper provides necessary and sufficient conditions for I-O constraints to be satisfied simultaneously in a class of closed-loop LTI systems with integral action, for command tracking, using a CBF-based governor that provides a suitably modified command. These conditions also imply that for some  MIMO plants with state feedback and integral action, it may be impossible to ensure I-O constraint satisfaction. It is also shown in this paper that  the SIOCF condition can always be satisfied
by expanding the design space where the input constraints are relaxed.
A systematic design of the governor parameters is proposed that determines a near-minimal relaxation of the input constraints. Overall,
this paper provides an I-O CBF Governor for modifying the command for a class of LTI plants that guarantees that necessary and sufficient conditions are always met in ensuring a simultaneous satisfaction of I-O constraints.

\bibliographystyle{IEEEtran}
\bibliography{refs}

\appendix
\section{}

Throughout the Appendix, $b(\xb)=(b_{z+};b_{z-};b_{u+};b_{u-};b_{M+};b_{M-})$ is partitioned conformally with the six blocks of $A_\zeta$ in \eqref{eq:qpdata}.
For $a\in\R^{m}$, $a_+$ and $a_-$ denote its componentwise positive and negative parts, $(a_+)_i:=\max\{a_i,0\}$ and $(a_-)_i:=\max\{-a_i,0\}$, so that $a=a_+-a_-$, and $\lvert a\rvert$ is taken componentwise.
For $\lambda\in\R^{6m}$, $\operatorname{supp}\lambda$ denotes the set of indices of its nonzero entries, and for a set $S$ of such indices, $A_S$ denotes the submatrix of $A_\zeta$ formed by the rows with indices in $S$ and $\lambda|_S$ the corresponding entries of $\lambda$.

\subsection{Proof of \texorpdfstring{\cref{prop:cert}}{Proposition \ref{prop:cert}}}

\begin{proof}
We use Farkas' lemma \cite{Schrijver86}: for $A\in\R^{p\times n}$ and $b\in\R^{p}$,
\begin{gather*}
\{\zeta\in\R^{n}:\ A\zeta\le b\}\neq\varnothing\\
\iff\ \lambda\T b\ge0\quad\forall\lambda\in\{\lambda\in\R^{p}:\ \lambda\ge0,\ A\T\lambda=0\} .
\end{gather*}
With $A=A_\zeta$ and $b=b(\xb)$, the set on the left is $F(\xb)$ by \eqref{eq:qpF} and the set on the right is $\Lambda$ by \eqref{eq:cone}, and the lemma directly yields
\[
F(\xb)\neq\varnothing
\iff
\lambda\T b(\xb)\ge0\quad\forall\lambda\in\Lambda .
\]

\emph{($\Rightarrow$)} Let $\xb\in\Sset$ and $\zeta\in F(\xb)$, that is, $A_\zeta\zeta\le b(\xb)$ by \eqref{eq:qpF}.
For every $k$, since $\lambda^k\ge0$ and $A_\zeta\T\lambda^k=0$ by \eqref{eq:cone},
\[
\lambda^{k\top}b(\xb)\ \ge\ \lambda^{k\top}A_\zeta\zeta\ =\ (A_\zeta\T\lambda^k)\T\zeta\ =\ 0 .
\]
Since $\xb\in\Sset$ is arbitrary, taking the infimum over $\Sset$ yields
\[
\nu_k(\Sset)=\inf_{\xb\in\Sset}\lambda^{k\top}b(\xb)\ \ge\ 0\qquad\forall k=1,\dots,N ,
\]
that is, \eqref{eq:siocfcond}.

\emph{($\Leftarrow$)} By contraposition.
Suppose $F(\xb)=\varnothing$ for some $\xb\in\Sset$.
By Farkas' lemma, there exists $\lambda\in\Lambda$ with $\lambda\T b(\xb)<0$.
By \eqref{eq:rays}, $\lambda=\sum_{k=1}^{N}c_k\lambda^k$ with $c_k\ge0$, and hence
\[
0\ >\ \lambda\T b(\xb)\ =\ \sum_{k=1}^{N}c_k\,\lambda^{k\top}b(\xb) ,
\]
which implies $\lambda^{k\top}b(\xb)<0$ for some $k$ with $c_k>0$.
Therefore
\[
\nu_k(\Sset)=\inf_{\xb'\in\Sset}\lambda^{k\top}b(\xb')\ \le\ \lambda^{k\top}b(\xb)\ <\ 0 ,
\]
which violates \eqref{eq:siocfcond}.
\end{proof}

\subsection{Proof of the bound on $N$ in \texorpdfstring{\cref{prop:cert}}{Proposition \ref{prop:cert}}}\label{app:Nbound}

\begin{proof}
For $\lambda\in\Lambda$ let $S:=\operatorname{supp}\lambda$, so that $\lambda|_S>0$.
Since $\lambda$ vanishes outside $S$, $0=A_\zeta\T\lambda=A_S\T(\lambda|_S)$, that is, $\lambda|_S\in\ker(A_S\T)$.

\emph{Step 1.} We show that, for a nonzero $\lambda\in\Lambda$,
\begin{equation}\label{eq:extray}
\lambda\ \text{spans an extreme ray of }\Lambda\iff\dim\ker(A_S\T)=1 .
\end{equation}
If $\dim\ker(A_S\T)\ge2$, choose $w\in\ker(A_S\T)$ independent of $\lambda|_S$; since $\lambda|_S>0$, for small $\epsilon>0$ the vectors supported on $S$ with values $\lambda|_S\pm\epsilon w$ lie in $\Lambda$, and $\lambda$ is their midpoint.
Conversely, if $\dim\ker(A_S\T)=1$ and $\lambda=\lambda'+\lambda''$ with $\lambda',\lambda''\in\Lambda$, both summands vanish off $S$, and their restrictions lie in $\ker(A_S\T)$ and are multiples of $\lambda$.
Hence $\operatorname{rank}A_S=\lvert S\rvert-1\le m$, and by \textnormal{\ref{ass:A1}} and \textnormal{\ref{ass:A3}} $\Hu K_I$ and $K_I$ are nonsingular, which excludes $\lvert S\rvert=1$.
Therefore $2\le\lvert S\rvert\le m+1$.

\emph{Step 2.} The $6m$ columns of $A_\zeta\T$ form $3m$ antipodal pairs $\{a_\ell,-a_\ell\}$, $\ell=1,\dots,3m$, where $a_\ell$ ranges over the columns of $(\Hu K_I)\T$, $K_I\T$ and $I_m$.
By \arxref{eq:extray}, each pair supports an extreme ray, its kernel being spanned by $[1\ \ 1]\T$, and a support containing both indices of a pair is exactly that pair, since otherwise $\lambda|_S$ and the vector equal to $[1\ \ 1]\T$ on the pair and $0$ elsewhere are two independent elements of $\ker(A_S\T)$.
With $\lVert\lambda^k\rVert_2=1$ fixing the scalar, the pairs contribute exactly $3m$ extreme rays.
Six more families are available for every plant and gain, two for each of the block pairs $(\Hu K_I,K_I)$, $(\Hu K_I,I_m)$ and $(K_I,I_m)$, obtained by expanding one column of the second block of $A_\zeta\T$ in the columns of the first.
For $i=1,\dots,m$ let
\begin{gather*}
p^{i}:=(\Hu\T)^{-1}e_i,\qquad q^{i}:=\big((\Hu K_I)\T\big)^{-1}e_i,\\
r^{i}:=(K_I\T)^{-1}e_i ,
\end{gather*}
with $p^{i}_{\pm}$, $q^{i}_{\pm}$, $r^{i}_{\pm}$ their componentwise positive and negative parts, and take the nonnegative vectors supported as
\begin{gather*}
(p^{i}_{+};p^{i}_{-};0;e_i;0;0),\qquad (q^{i}_{+};q^{i}_{-};0;0;0;e_i),\\
(0;0;r^{i}_{+};r^{i}_{-};0;e_i),
\end{gather*}
together with their mirror images $(p^{i}_{-};p^{i}_{+};e_i;0;0;0)$, $(q^{i}_{-};q^{i}_{+};0;0;e_i;0)$ and $(0;0;r^{i}_{-};r^{i}_{+};e_i;0)$.
By the blocks of $A_\zeta$ in \eqref{eq:qpdata} all six lie in $\Lambda$, since $(\Hu K_I)\T p^{i}=K_I\T e_i$, $(\Hu K_I)\T q^{i}=e_i$ and $K_I\T r^{i}=e_i$; each identity is the only dependence among the columns of $A_\zeta\T$ that the vector selects, since the columns of each block form a basis of $\R^m$ and the expansion of a vector in a basis is unique; hence all of them are extreme by \arxref{eq:extray}.
Each support meets exactly two of the three blocks of \eqref{eq:qpdata}, the second of them at the single index $i$ with one sign.
The six families are therefore distinct from the pairs, from one another, and within each family across $i$ and across the mirror; hence $N\ge3m+6m=9m$.

\emph{Step 3.} Every other support $S$ meets each pair at most once and is therefore of the form $S=\{(\ell,\varsigma_\ell)\}_{\ell\in L}$ for some $L\subseteq\{1,\dots,3m\}$ with $2\le\lvert L\rvert\le m+1$ and some $\varsigma:L\to\{+,-\}$.
By \arxref{eq:extray}, $\operatorname{rank}\{a_\ell\}_{\ell\in L}=\lvert L\rvert-1$, a condition on $L$ alone, and the dependence $\sum_{\ell\in L}c_\ell a_\ell=0$ is then unique up to a nonzero scalar, with every $c_\ell\neq0$; hence at most two signings, $\pm\operatorname{sign}(c)$, admit a positive kernel vector.
No such $L$ contains another, since a dependence supported on a proper subset would be independent of the one supported on $L$.
Counting the pairs $(\pi,L)$ with $\pi$ a permutation of $\{1,\dots,3m\}$ whose first $\lvert L\rvert$ entries enumerate $L$, each $L$ occurs in $\lvert L\rvert!\,(3m-\lvert L\rvert)!$ pairs and each $\pi$ at most once, which yields $\sum_L\binom{3m}{\lvert L\rvert}^{-1}\le1$.
Since $2\le\lvert L\rvert\le m+1\le\tfrac{3m}{2}$ for $m\ge2$, and $\lvert L\rvert=m+1$ is forced when $m=1$, each binomial coefficient is at most $\binom{3m}{m+1}$; hence there are at most $\binom{3m}{m+1}$ such sets $L$, each giving at most two extreme rays, and $N\le3m+2\binom{3m}{m+1}$.
\end{proof}

\subsection{Proof of \texorpdfstring{\cref{prop:escape}}{Proposition \ref{prop:escape}}}

\begin{proof}
We first show that the condition
\begin{equation}\label{eq:vanish}
\Cz\Ab^{2}v=0\quad\text{and}\quad\Kfull\Ab v=0
\qquad\forall v\in\ker Q
\end{equation}
is equivalent to $\ker Q=\Oh$.
Every $v\in\ker Q$ satisfies $\Cz v=\Cz\Ab v=\Kfull v=0$ by the definition of $Q$ in \eqref{eq:omegamat}.

($\Rightarrow$) Assume \eqref{eq:vanish} and let $v\in\ker Q$.
Then $\Cz(\Ab v)=0$, $\Cz\Ab(\Ab v)=0$ and $\Kfull(\Ab v)=0$, that is, $\Ab v\in\ker Q$.
By induction, $\Ab^{k}v\in\ker Q$ for all $k\ge0$.
Since $w\in\ker Q$ implies $w\in\ker\Cbar$, $\Cbar\Ab^{k}v=0$ for all $k\ge0$, that is, $v\in\Oh$ by \arxref{eq:unobs}, and $\ker Q\subseteq\Oh$.
Conversely, let $v\in\Oh$.
The conditions $\Cbar v=0$ and $\Cbar\Ab v=0$ contain $\Cz v=0$, $\Kfull v=0$ and $\Cz\Ab v=0$, that is, $v\in\ker Q$.
Hence $\Oh\subseteq\ker Q$, and therefore $\ker Q=\Oh$.

($\Leftarrow$) Assume $\ker Q=\Oh$ and let $v\in\ker Q$.
Then $v\in\Oh$, and the conditions $\Cbar\Ab^{2}v=0$ and $\Cbar\Ab v=0$ contain $\Cz\Ab^{2}v=0$ and $\Kfull\Ab v=0$, which is \eqref{eq:vanish}.

Suppose $\ker Q\neq\Oh$. Then $\Cz\Ab^{2}v\neq0$ or $\Kfull\Ab v\neq0$ for some $v\in\ker Q$.

\emph{Case (i): $\Cz\Ab^{2}v\neq0$.}

\emph{Step 1.} Define $R:=\Cz\Ab$, $g:=\Cz\Ab^{2}v$, and
\[
\beta(\xb):=g\T R\,\xb,\qquad \gamma:=g\T R\Ab v=\lVert\Cz\Ab^{2}v\rVert^{2}>0 .
\]
Since $v\in\ker Q$, $Rv=\Cz\Ab v=0$.
On $\Omega$, the inequalities $\psi_1^{\pm}\ge0$ and $h_1^{\pm}\ge0$ in \eqref{eq:omega} imply, by \eqref{eq:psi1},
\begin{equation}\label{eq:Rbound}
\lvert R\xb\rvert\le\rho:=\alpha_1(\zmax-\zmin) ,
\end{equation}
Define $B:=\lvert g\rvert\T\rho$.
Then, for every $\xb\in\Omega$,
\begin{equation}\label{eq:betaB}
\lvert\beta(\xb)\rvert=\lvert g\T R\xb\rvert\le\lvert g\rvert\T\lvert R\xb\rvert\le B .
\end{equation}

\emph{Step 2.} Without loss of generality $t_0=0$.
Let $\xb_0:=rv$ with $r>0$.
Since $v\in\ker Q$, $\xb_0\in\Omega$.
Every solution exists globally and is given by
\[
\xb(t)=e^{\Ab t}rv+\int_0^{t}e^{\Ab(t-\tau)}\By\,\zeta(\tau)\,d\tau .
\]
Substituting it into $\beta$ splits $\beta(\xb(t))=r\,f(t)+\varphi(t)$ into the response to $\xb_0$ and to $\zeta$, where
\begin{gather*}
f(t):=g\T Re^{\Ab t}v ,\qquad
\varphi(t):=g\T R\int_0^{t}e^{\Ab(t-\tau)}\By\,\zeta(\tau)\,d\tau .
\end{gather*}
Define
\[
c:=\sqrt m\,\lVert R\T g\rVert\,\lVert\By\rVert\max_{t\in[0,\bar\tau]}\lVert e^{\Ab t}\rVert .
\]
From $\lVert\zeta\rVert_\infty\le M$ and $\lVert\zeta\rVert\le\sqrt m\,\lVert\zeta\rVert_\infty$,
\[
\lvert\varphi(t)\rvert\le t\,cM\qquad\forall t\in[0,\bar\tau] .
\]
From $Rv=0$, $f(0)=0$ and $f'(0)=g\T R\Ab v=\gamma$.
Moreover, $f''(t)=g\T R\Ab^{2}e^{\Ab t}v$ is bounded on $[0,\bar\tau]$ by $\kappa:=\max_{t\in[0,\bar\tau]}\lvert g\T R\Ab^{2}e^{\Ab t}v\rvert$.
Taylor's theorem then yields
\begin{equation}\label{eq:ftaylor}
f(t)\ge\gamma t-\tfrac{\kappa}{2}t^{2}\qquad\forall t\in[0,\bar\tau] .
\end{equation}

\emph{Step 3.} Define $t^\ast:=\min\{\bar\tau,\ \gamma/\kappa\}$, with $\gamma/\kappa:=\infty$ if $\kappa=0$, and $r^\ast:=2(B+cMt^\ast)/(\gamma t^\ast)$.
Since $t^\ast\le\gamma/\kappa$, \arxref{eq:ftaylor} yields
\[
f(t^\ast)\ge\gamma t^\ast-\tfrac{\kappa}{2}t^{\ast2}\ge\tfrac{\gamma}{2}t^\ast .
\]
For every $r>r^\ast$ and every piecewise continuous $\zeta$ with $\lVert\zeta\rVert_\infty\le M$, it follows that
\[
\beta(\xb(t^\ast))\ \ge\ r\tfrac{\gamma}{2}t^\ast-cMt^\ast\ >\ B ,
\]
and therefore $\xb(t^\ast)\notin\Omega$ by \arxref{eq:betaB}.

\emph{Case (ii): $\Kfull\Ab v\neq0$.}
The proof is the same, with $\Cz\Ab$ replaced by $\Kfull$ and $\rho$ in \arxref{eq:Rbound} by $\umax-\umin$.
\end{proof}

\subsection{Proof of \texorpdfstring{\cref{thm:fin}}{Proposition \ref{thm:fin}}}

\begin{proof}
Define
\[
p:=\lambda_{z+}-\lambda_{z-},\quad q:=\lambda_{u+}-\lambda_{u-},\quad s:=\lambda_{M+}-\lambda_{M-} ,
\]
so that, by \eqref{eq:qpdata},
\begin{equation}\label{eq:AzetaT}
A_\zeta\T\lambda=(\Hu K_I)\T p+K_I\T q+s .
\end{equation}
The map $\lambda\mapsto(p,q)$ is a surjection of $\Lambda$ onto $\R^{2m}$.
Given $(p,q)\in\R^{2m}$, define
\[
s:=-\big((\Hu K_I)\T p+K_I\T q\big),\ \lambda:=(p_+;p_-;q_+;q_-;s_+;s_-) .
\]
Then $\lambda\ge0$ and $A_\zeta\T\lambda=0$ by \eqref{eq:AzetaT}, that is, $\lambda\in\Lambda$.

Decompose $\xb=x_r+v$ with $x_r\in(\ker Q)^{\perp}=\operatorname{row}(Q)$ and $v\in\ker Q$, so that $\lambda^{k\top}b(\xb)=\lambda^{k\top}(\Bz x_r+b_0)+\lambda^{k\top}\Bz v$.
Then
\begin{equation}\label{eq:finiteness}
\nu_k\big(\Omega\big)>-\infty \iff \lambda^{k\top}\Bz v=0\quad\forall\,v\in\ker Q .
\end{equation}
By \eqref{eq:finiteness}, it suffices to show that $\lambda^{k\top}\Bz v=0$ for every $k$ and every $v\in\ker Q$ if and only if $\ker Q=\Oh$.
Let $v\in\ker Q$.
By \eqref{eq:qpdata}, for every $(\alpha_1,\alpha_2,\alpha_u,\ag)$,
\[
\Bz v=(-\Cz\Ab^{2}v;\ \Cz\Ab^{2}v;\ \Kfull\Ab v;\ -\Kfull\Ab v;\ 0;\ 0) ,
\]
so that, by the definitions of $p$ and $q$,
\begin{equation}\label{eq:lamB}
\lambda\T\Bz v=-p\T\Cz\Ab^{2}v+q\T\Kfull\Ab v .
\end{equation}

\emph{($\Rightarrow$)} Suppose $\nu_k(\Omega)>-\infty$ for every $k$.
By \eqref{eq:finiteness}, $\lambda^{k\top}\Bz v=0$ for every $k$ and every $v\in\ker Q$.
Since every $\lambda\in\Lambda$ is a nonnegative combination of the extreme rays by \eqref{eq:rays},
\[
\lambda\T\Bz v=0\qquad\forall\lambda\in\Lambda,\ \forall v\in\ker Q .
\]
By \eqref{eq:lamB} and the surjectivity of $\lambda\mapsto(p,q)$,
\[
-p\T\Cz\Ab^{2}v+q\T\Kfull\Ab v=0\qquad\forall(p,q)\in\R^{2m} .
\]
Taking $q=0$ and $p=\Cz\Ab^{2}v$ yields $\Cz\Ab^{2}v=0$, and taking $p=0$ and $q=\Kfull\Ab v$ yields $\Kfull\Ab v=0$.
Since $v\in\ker Q$ is arbitrary, \eqref{eq:vanish} is satisfied, that is, $\ker Q=\Oh$.

\emph{($\Leftarrow$)} Suppose $\ker Q=\Oh$.
Then \eqref{eq:vanish} is satisfied, and \eqref{eq:lamB} yields $\lambda^{k\top}\Bz v=0$ for every $k$ and every $v\in\ker Q$.
By \eqref{eq:finiteness}, $\nu_k(\Omega)>-\infty$ for every $k$.
\end{proof}

\subsection{Proof of \texorpdfstring{\cref{prop:kerQO}}{Proposition \ref{prop:kerQO}}}

\begin{proof}
Split $v=\begin{bsmallmatrix}v_p\\ v_e\end{bsmallmatrix}$ with $v_p\in\R^{n_p}$, $v_e\in\R^{m}$, and let $\phi(w):=\begin{bsmallmatrix}w\\ -K_I^{-1}K_Pw\end{bsmallmatrix}$ for $w\in\R^{n_p}$, a linear injection.

\emph{Step 1.} First, we will show that
\begin{equation}\label{eq:kerQphi}
\ker Q=\phi(\operatorname{span}V) .
\end{equation}
From \eqref{eq:loop}, $\Cz\Ab+\Hu\Kfull=[\,\Clim A_p\ \ 0\,]$.
With the nonsingular $T_H:=\begin{bsmallmatrix}I&0&0\\ 0&I&\Hu\\ 0&0&I\end{bsmallmatrix}$, this identity yields
\[
T_HQ=\begin{bmatrix}\Clim&0\\ \Clim A_p&0\\ K_P&K_I\end{bmatrix},\qquad \ker Q=\ker(T_HQ) ,
\]
and $v\in\ker Q$ if and only if
\[
\Clim v_p=0,\qquad \Clim A_pv_p=0,\qquad K_Pv_p+K_Iv_e=0 .
\]
The first two conditions are equivalent to $v_p\in\operatorname{span}V$.
Under \textnormal{\ref{ass:A3}}, $K_I$ is nonsingular: if $K_Iw=0$ with $w\neq0$, then $\Ab\begin{bsmallmatrix}0\\ w\end{bsmallmatrix}=0$ by \eqref{eq:loop}, which contradicts $\Ab$ being Hurwitz.
The third condition is equivalent to $v_e=-K_I^{-1}K_Pv_p$, which proves \arxref{eq:kerQphi}.

\emph{Step 2.} Next, we will show that
\begin{equation}\label{eq:Ophi}
\Oh=\phi(\Nt) ,
\end{equation}
where $\Nt:=\{w\in\R^{n_p}:\ C_NA_p^{k}w=0\ \forall k\ge0\}$ with $C_N:=\begin{bsmallmatrix}\Clim\\ K_PA_p+K_I\Creg\end{bsmallmatrix}$, which is the unobservable subspace of the pair $(A_p,C_N)$.
To prove \arxref{eq:Ophi}, we first show that, for $k\ge0$,
\begin{equation}\label{eq:rec}
\begin{gathered}
\Kfull\Ab^{i}v=0,\ \ i=0,\dots,k\ \Longrightarrow\\
\Ab^{j}v=\begin{bmatrix}A_p^{j}v_p\\ \Creg A_p^{j-1}v_p\end{bmatrix},\ \ j=1,\dots,k+1 .
\end{gathered}
\end{equation}
Assume that $\Kfull\Ab^{i}v=0$ for $i=0,\dots,k$.
Then \eqref{eq:loop} and $\Kfull v=0$ yield
\[
\Ab v=\begin{bmatrix}A_pv_p\\ \Creg v_p\end{bmatrix}-\begin{bmatrix}B_p\\ \Dreg\end{bmatrix}\Kfull v=\begin{bmatrix}A_pv_p\\ \Creg v_p\end{bmatrix} ,
\]
which is the claim for $j=1$.
If $\Ab^{j}v=\begin{bsmallmatrix}A_p^{j}v_p\\ \Creg A_p^{j-1}v_p\end{bsmallmatrix}$ for some $j\in\{1,\dots,k\}$, then \eqref{eq:loop} and $\Kfull\Ab^{j}v=0$ yield
\[
\Ab^{j+1}v=\begin{bmatrix}A_p^{j+1}v_p\\ \Creg A_p^{j}v_p\end{bmatrix}-\begin{bmatrix}B_p\\ \Dreg\end{bmatrix}\Kfull\Ab^{j}v=\begin{bmatrix}A_p^{j+1}v_p\\ \Creg A_p^{j}v_p\end{bmatrix} ,
\]
which is the claim for $j+1$.
By induction on $j$, this proves \arxref{eq:rec}.
We now prove the inclusions $\Oh\subseteq\phi(\Nt)$ and $\phi(\Nt)\subseteq\Oh$.

($\subseteq$) Let $v\in\Oh$, that is, $\Cz\Ab^{k}v=0$ and $\Kfull\Ab^{k}v=0$ for every $k\ge0$.
At $k=0$ these are $\Clim v_p=0$ and $K_Pv_p+K_Iv_e=0$, the latter being $v=\phi(v_p)$.
Since $\Kfull\Ab^{k}v=0$ for every $k\ge0$, \arxref{eq:rec} applies for every $k$ and yields, for $k\ge1$,
\begin{gather*}
\Cz\Ab^{k}v=\Clim A_p^{k}v_p=0,\\
\Kfull\Ab^{k}v=(K_PA_p{+}K_I\Creg)A_p^{k-1}v_p=0 .
\end{gather*}
Together with $\Clim v_p=0$, these yield $C_NA_p^{k}v_p=0$ for every $k\ge0$, that is, $v_p\in\Nt$.
Hence $\Oh\subseteq\phi(\Nt)$.

($\supseteq$) Let $w\in\Nt$ and $v=\phi(w)$, that is, $v_p=w$ and $v_e=-K_I^{-1}K_Pw$.
We show by induction on $k$ that $\Kfull\Ab^{k}v=0$ for every $k\ge0$.
For $k=0$, $\Kfull v=K_Pw-K_IK_I^{-1}K_Pw=0$.
If $\Kfull\Ab^{i}v=0$ for $i=0,\dots,k$, then \arxref{eq:rec} with $j=k+1$ yields
\begin{gather*}
\Ab^{k+1}v=\begin{bmatrix}A_p^{k+1}w\\ \Creg A_p^{k}w\end{bmatrix},\\
\Kfull\Ab^{k+1}v=(K_PA_p{+}K_I\Creg)A_p^{k}w=0 ,
\end{gather*}
the last equality by $w\in\Nt$.
Then \arxref{eq:rec} and $w\in\Nt$ yield $\Cz\Ab^{k}v=\Clim A_p^{k}w=0$ for every $k\ge0$.
Hence $\Cbar\Ab^{k}v=0$ for every $k\ge0$, that is, $v\in\Oh$, and $\phi(\Nt)\subseteq\Oh$.

\emph{Step 3.} By \arxref{eq:kerQphi}, \arxref{eq:Ophi} and the injectivity of $\phi$, $\ker Q=\Oh$ if and only if $\operatorname{span}V=\Nt$.
Every $w\in\Nt$ satisfies $\Clim w=0$ and $\Clim A_pw=0$ by the definition of $C_N$, that is, $\Nt\subseteq\operatorname{span}V$, and $\operatorname{span}V=\Nt$ reduces to $\operatorname{span}V\subseteq\Nt$.
We will show that $\operatorname{span}V\subseteq\Nt$ is equivalent to
\begin{gather*}
\operatorname{span}(A_pV)\subseteq\operatorname{span}V\quad\text{and}\\
(K_PA_p{+}K_I\Creg)w=0\ \ \forall w\in\operatorname{span}V .
\end{gather*}
($\Rightarrow$) Let $\operatorname{span}V\subseteq\Nt$ and $w\in\operatorname{span}V$, which implies $w\in\Nt$.
The condition defining $\Nt$ at the index $k=0$ yields $(K_PA_p{+}K_I\Creg)w=0$.
At the index $k+1$, it yields $C_NA_p^{k}(A_pw)=0$ for every $k\ge0$, that is, $A_pw\in\Nt\subseteq\operatorname{span}V$.

($\Leftarrow$) Suppose that $\operatorname{span}(A_pV)\subseteq\operatorname{span}V$ and $(K_PA_p{+}K_I\Creg)w=0$ for every $w\in\operatorname{span}V$.
Then $A_p^{k}w\in\operatorname{span}V$ for every $w\in\operatorname{span}V$ and every $k\ge0$; hence $C_NA_p^{k}w=0$, that is, $w\in\Nt$.

Since $\operatorname{span}V=\ker\begin{bsmallmatrix}\Clim\\ \Clim A_p\end{bsmallmatrix}$ and $\Clim A_pV=0$,
\begin{gather*}
\operatorname{span}(A_pV)\subseteq\operatorname{span}V\\
\Leftrightarrow\begin{bmatrix}\Clim\\ \Clim A_p\end{bmatrix}A_pV=0\Leftrightarrow\Clim A_p^{2}V=0,\\
(K_PA_p{+}K_I\Creg)w=0\ \ \forall w\in\operatorname{span}V\\
\iff(K_PA_p{+}K_I\Creg)V=0 ,
\end{gather*}
which together prove \cref{prop:kerQO}.
\end{proof}

\subsection{Proof of \texorpdfstring{\cref{lem:normalform}}{Lemma \ref{lem:normalform}}}

\begin{proof}
\emph{(a)} By \textnormal{\ref{ass:A1}}, $\Hu=\Clim B_p$ is nonsingular, which implies $\operatorname{rank}B_p=m$.
Since $\operatorname{span}W_p\T=\ker(B_p\T)$ has dimension $n_p-m$, $W_p\T$ has full column rank, and therefore $W_pB_p=0$ and $W_pW_p\T$ is nonsingular.
Moreover, $M:=[B_p\ \ W_p\T]$ is nonsingular, since $\operatorname{span}B_p$ and $\operatorname{span}W_p\T=\ker(B_p\T)$ are orthogonal complements in $\R^{n_p}$.
Stacking $\Clim$ and $W_p$ in $T$ yields
\[
TM=\begin{bmatrix}\Clim B_p&\Clim W_p\T\\ W_pB_p&W_pW_p\T\end{bmatrix}=\begin{bmatrix}\Hu&\Clim W_p\T\\ 0&W_pW_p\T\end{bmatrix} ,
\]
which is block upper triangular with nonsingular diagonal blocks.
Hence $TM$, and therefore $T$, is nonsingular.

\emph{(b)} By the partition $T^{-1}=:[R_1\ \ R_2]$ as in (b),
\[
TT^{-1}=\begin{bmatrix}\Clim\\ W_p\end{bmatrix}[R_1\ \ R_2]
=\begin{bmatrix}\Clim R_1&\Clim R_2\\ W_pR_1&W_pR_2\end{bmatrix}
=I_{n_p} .
\]
Differentiating $\etat=W_px_p$ along \eqref{eq:plant} and using $W_pB_p=0$,
\[
\dot\etat=W_p(A_px_p+B_pu)=W_pA_px_p .
\]
Since $Tx_p=(z,\etat)$ and $T^{-1}=[R_1\ \ R_2]$, $x_p=R_1z+R_2\etat$, and substituting it yields $\dot\etat=W_pA_pR_1z+W_pA_pR_2\etat=B_\eta z+A_\eta\etat$.

\emph{(c)} By \textnormal{\ref{ass:A2}}, the invariant zeros of $(A_p,B_p,\Clim)$ \cite{Isidori95} lie in $\mathbb C^{-}$.
They are the roots of $\det P(\lambda)$, where
\[
P(\lambda):=\begin{bmatrix}A_p-\lambda I& B_p\\ \Clim& 0\end{bmatrix} .
\]
Since $\det\operatorname{diag}(T,I_m)\cdot\det\operatorname{diag}(T^{-1},I_m)=1$, the matrix $\tilde P(\lambda):=\operatorname{diag}(T,I_m)\,P(\lambda)\,\operatorname{diag}(T^{-1},I_m)$, which is $P(\lambda)$ written in the coordinates $(z,\etat)$ of (b), satisfies $\det\tilde P(\lambda)=\det P(\lambda)$, and substituting (b) yields
\[
\tilde P(\lambda)=\begin{bmatrix}\ast&\ast&\Hu\\ B_\eta&A_\eta-\lambda I&0\\ I_m&0&0\end{bmatrix} ,
\]
whose determinant is $\pm\det(\Hu)\det(A_\eta-\lambda I)$.
The invariant zeros are therefore the eigenvalues of $A_\eta$, and $A_\eta$ is Hurwitz.

\emph{(d)} $A_\eta$ is nonsingular by (c), and $\Bd$ is well defined.
Differentiating $\eta=\etat+\Bd z$ and using (b) with $A_\eta\Bd=B_\eta$ yields \eqref{eq:etadyn}.

\emph{(e)} Since $\Tz\in\R^{(n_p+m)\times(n_p+m)}$ is square, it is nonsingular if and only if $\ker\Tz=\{0\}$; as $0\in\ker\Tz$ always, it suffices to show that $\Tz v=0$ implies $v=0$.
Let $v=\begin{bsmallmatrix}v_p\\ v_e\end{bsmallmatrix}$ with $v_p\in\R^{n_p}$ and $v_e\in\R^{m}$, and suppose that $\Tz v=0$.
Then
\[
\Cz v=0,\qquad \Cz\Ab v=0,\qquad [\,\Wt\ \ 0\,]\,v=0 .
\]
The first and third equations yield $\Clim v_p=0$ and $\Wt v_p=0$, and hence $W_pv_p=0$, that is, $Tv_p=0$, which implies $v_p=0$ by (a).
The second equation then yields $0=\Cz\Ab\begin{bsmallmatrix}0\\ v_e\end{bsmallmatrix}=-\Hu K_Iv_e$, which implies $v_e=0$ by \textnormal{\ref{ass:A1}} and \textnormal{\ref{ass:A3}}.
\end{proof}

\subsection{Proof that $E$ of Section~\ref{sec:rpi} is RPI}

\begin{proof}
For $s\ge0$ and piecewise continuous $w:[0,s]\to\W$, let $\Phi_s^{w}:\R^{n_p-m}\to\R^{n_p-m}$ denote the solution map of \eqref{eq:etadyn} over $[0,s]$ under the input $w$,
\[
\Phi_s^{w}(\eta):=e^{A_\eta s}\eta+\int_0^{s}e^{A_\eta(s-\tau)}\Bd\,w(\tau)\,d\tau .
\]
By the construction of $E$ in \arxref{eq:Eset}, $E=\operatorname{cl}\{\operatorname{conv}\{R\}\}$, where
\[
R:=\{\Phi_{r}^{w}(\eta(t_0)):\ r\ge0,\ w:[0,r]\to\W\} ,
\]
with $w$ piecewise continuous, is the set of states that \eqref{eq:etadyn} reaches from $\eta(t_0)$.

\emph{Step 1.} Since $A_\eta$ is Hurwitz by \cref{lem:normalform}\ref{lem:nf-c}, $\lVert e^{A_\eta t}\rVert\le c\,e^{-\sigma t}$ for some $c,\sigma>0$ and $\lVert e^{A_\eta t}\eta(t_0)\rVert\le c\lVert\eta(t_0)\rVert$; the integral is bounded because $\W$ is bounded.
Hence $R$ is bounded, as are $\operatorname{conv}\{R\}$ and its closure, which is closed by construction and convex as the closure of a convex set.
Thus $E$ is compact and convex.
Taking $w\equiv0$ shows that $e^{A_\eta r}\eta(t_0)\in R$ for every $r\ge0$; since this point tends to $0$ and $E$ is closed, $0\in E$.

\emph{Step 2.} We show that, for every $s\ge0$ and piecewise continuous $w_2:[0,s]\to\W$,
\begin{equation}\label{eq:PhiR}
\Phi_s^{w_2}(R)\subseteq R .
\end{equation}
Let $\eta\in R$, that is, $\eta=\Phi_{r}^{w_1}(\eta(t_0))$ for some $r\ge0$ and piecewise continuous $w_1:[0,r]\to\W$.
Define $w_3:[0,r+s]\to\W$ by $w_3(\tau):=w_1(\tau)$ for $\tau\in[0,r]$ and $w_3(\tau):=w_2(\tau-r)$ for $\tau\in(r,r+s]$, which is piecewise continuous.
Splitting the integral in $\Phi_{r+s}^{w_3}$ at $\tau=r$ yields
\[
\Phi_{r+s}^{w_3}(\eta(t_0))=\Phi_s^{w_2}\big(\Phi_{r}^{w_1}(\eta(t_0))\big)=\Phi_s^{w_2}(\eta) ,
\]
and hence $\Phi_s^{w_2}(\eta)\in R$.

\emph{Step 3.} From the property of the convex hull, every $p\in\operatorname{conv}\{R\}$ is of the form $p=\sum_{j=1}^{J}\theta_j\eta_j$ with $J\ge1$, $\eta_j\in R$, $\theta_j\ge0$ and $\sum_{j=1}^{J}\theta_j=1$.
Since $\Phi_s^{w_2}(\eta)=e^{A_\eta s}\eta+\Phi_s^{w_2}(0)$ is affine in $\eta$ and $\sum_{j=1}^{J}\theta_j=1$,
\[
\Phi_s^{w_2}\Big(\sum_{j=1}^{J}\theta_j\eta_j\Big)=\sum_{j=1}^{J}\theta_je^{A_\eta s}\eta_j+\Phi_s^{w_2}(0)=\sum_{j=1}^{J}\theta_j\Phi_s^{w_2}(\eta_j) ,
\]
which belongs to $\operatorname{conv}\{R\}$ by \arxref{eq:PhiR}.
Since $\Phi_s^{w_2}$ is continuous, $\Phi_s^{w_2}(E)\subseteq\operatorname{cl}\{\Phi_s^{w_2}(\operatorname{conv}\{R\})\}\subseteq E$, which, together with Step 1, proves that $E$ satisfies \cref{def:rpi}.

\emph{Step 4.} Each integral in \eqref{eq:emin}, written with an input $w_1:[0,t]\to\W$, equals $\Phi_t^{w_2}(0)$ with $w_2(\tau):=w_1(t-\tau)$.
For every $T\ge0$, $e^{A_\eta T}\eta(t_0)\in R$ and \arxref{eq:PhiR} yield $\Phi_t^{w_2}(e^{A_\eta T}\eta(t_0))=e^{A_\eta(t+T)}\eta(t_0)+\Phi_t^{w_2}(0)\in R$, which tends to $\Phi_t^{w_2}(0)$ as $T\to\infty$.
Since $E$ is closed, $\Emin\subseteq E$.
\end{proof}

\subsection{Proof of \texorpdfstring{\cref{prop:fix4nec}}{Proposition \ref{prop:fix4nec}}}

\begin{proof}
\emph{(a)} For all $u\in\U$, $\dot z=az+u\ge az-\umax$, that is, $\dot y\ge ay$ with $y:=z-\umax/a$.
The comparison lemma \cite[Lemma~3.4]{Khalil02} applied to this inequality yields
\[
z(t)-\cfrac{\umax}{a}\ \ge\ \Big(z(t_0)-\cfrac{\umax}{a}\Big)e^{a(t-t_0)} ,
\]
which is the bound in (a).
Since $a>0$ and $z(t_0)-\umax/a>0$, the lower bound equals $\zmax$ at
\[
t_1:=t_0+\cfrac{1}{a}\ln\cfrac{\zmax-\umax/a}{z(t_0)-\umax/a}\ \ge\ t_0
\]
and exceeds $\zmax$ for $t>t_1$; hence $z(t)>\zmax$ for $t>t_1$, and the output box is violated in finite time.

\emph{(b)} \textnormal{\ref{ass:A1}} is satisfied with $\Hu=1$, \textnormal{\ref{ass:A2}} is vacuous ($n_p=m$, no internal dynamics), and \textnormal{\ref{ass:A3}} implies $K_I\neq0$; $\ker\begin{bsmallmatrix}1\\ a\end{bsmallmatrix}=\{0\}$, which implies $d=0$ and $\SE=\Omega$.
In this setting, $\dot z=az+u$, and $h_1^{\pm}$, $\psi_1^{\pm}$ and $h_u^{\pm}$ depend on $\xb$ only through $z=\Cz\xb$ and $u=-\Kfull\xb$.

\emph{Step 1.} Let
\[
z^\ast:=\cfrac{\alpha_1\zmax+\umax}{a+\alpha_1} ,
\]
and let $\xb^\ast$ be the unique solution of $\Cz\xb^\ast=z^\ast$ and $-\Kfull\xb^\ast=-\umax$, which exists since $\begin{bsmallmatrix}\Cz\\ \Kfull\end{bsmallmatrix}=\begin{bsmallmatrix}1&0\\ K_P&K_I\end{bsmallmatrix}$ is nonsingular for $K_I\neq0$.
At $\xb^\ast$, $\dot z=az^\ast-\umax=\alpha_1(\zmax-z^\ast)$, which yields $\psi_1^{+}(\xb^\ast)=h_u^{-}(\xb^\ast)=0$, while $h_1^{+}(\xb^\ast)=(a\zmax-\umax)/(a+\alpha_1)>0$ and $h_1^{-},\psi_1^{-},h_u^{+}>0$ at $\xb^\ast$.
Therefore $\xb^\ast\in\Omega=\SE$ for every $\alpha_1>0$.

\emph{Step 2.} Let $\lambda:=(1;0;0;1;0;0)\in\R^{6}$, that is, $\lambda_{z+}=\lambda_{u-}=1$ and all other blocks of $\lambda$ vanish.
Since $\Hu=1$, \eqref{eq:qpdata} yields $A_\zeta\T\lambda=K_I-K_I=0$, and $\lambda/\sqrt2$ is an extreme ray of $\Lambda$ by \arxref{eq:extray}, since $\dim\ker(A_S\T)=1$ for its support $S$.
Hence $\lambda/\sqrt2=\lambda^{k}$ for some $k\in\{1,\dots,N\}$, and $\sigma_k=\lambda^{k\top}\eM=0$.

\emph{Step 3.} At $\xb^\ast$, $\lambda\T\Psi=\psi_1^{+}=0$ and $\lambda\T H=h_u^{-}=0$, and therefore $\lambda\T b=\lambda\T\hat b$ for every $(\alpha_2,\alpha_u,\ag)$.
By \eqref{eq:bsplit}, $\lambda\T\hat b(\xb)=-\Cz(\Ab+\alpha_1I)\Ab\xb-\Kfull\Ab\xb$.
Since $\dot z=az+u$ with $u=-\Kfull\xb$ and $\Cz\By=0$, $\Cz\Ab=a\Cz-\Kfull$, and hence $\Cz\Ab^{2}\xb=a\,\Cz\Ab\xb-\Kfull\Ab\xb$.
Substituting it yields
\[
\lambda\T\hat b(\xb)=-(a+\alpha_1)\,\Cz\Ab\xb ,
\]
which at $\xb^\ast$ equals $-(a+\alpha_1)\,\alpha_1(\zmax-z^\ast)=\alpha_1(\umax-a\zmax)$ by Step 1.
Hence, for every $\alpha^{1}\in\R_{++}^{4}$,
\begin{align*}
\min_k\nu_k(\SE)\le\nu_k&=\cfrac{1}{\sqrt2}\inf_{\SE}\lambda\T b\le\cfrac{1}{\sqrt2}\lambda\T b(\xb^\ast)\\
&=\cfrac{\alpha_1}{\sqrt2}(\umax-a\zmax)<0 ,
\end{align*}
which proves (b).
\end{proof}

\subsection{Proof of \texorpdfstring{\cref{prop:ceiling}}{Proposition \ref{prop:ceiling}}}

\begin{proof}
Fix $c>0$ and set $\alpha_u=c\,\alpha_2$.
Let $s_k:=\lambda^{k\top}(\Psi+cH)$, $\Kp:=\{k:\ \sigma_k>0\}$ and $\ThB_k:=\inf_{\xb\in Z_k}\lambda^{k\top}\hat b(\xb)$ with $\inf\varnothing:=+\infty$, so that $\ThB=\min_{k\in\Kz}\ThB_k$ by \eqref{eq:ThB}.
The set $\SE(\varrho)$ is compact by \cref{thm:mainSE}, and none of $\SE(\varrho)$, $\hat b$, $\Psi$ and $H$ depends on $(\alpha_2,\alpha_u,\ag)$, since $\SE(\varrho)$ is built from $h_1^{\pm}$, $\psi_1^{\pm}$, $h_u^{\pm}$ and $\eta$ (see \eqref{eq:omega} and \cref{def:rpi}).
Since $\Psi,H\ge0$ on $\SE(\varrho)\subseteq\Omega$ and $\lambda^k\ge0$, $s_k\ge0$ on $\SE(\varrho)$.
Each $Z_k$ is a closed subset of $\SE(\varrho)$ and $\lambda^{k\top}\hat b$ is continuous, so $\ThB_k$ is attained whenever $Z_k\neq\varnothing$, and $\ThB\in\R\cup\{+\infty\}$.
Substituting \eqref{eq:bsplit} into the definition of $\nu_k$ yields
\[
\nu_k(\SE(\varrho))=\inf_{\xb\in\SE(\varrho)}\big(\lambda^{k\top}\hat b(\xb)+\alpha_2\,s_k(\xb)\big)+\sigma_k\,\ag .
\]

We first show that, for every real $\varepsilon<\ThB$, $\min_k\nu_k(\SE(\varrho))\ge\varepsilon$ if and only if $\alpha_2\ge\underline\alpha_2$ and $\ag\ge\underline\alpha_g(\alpha_2)$, with
\begin{gather}
\underline\alpha_2:=\max_{k\in\Kz}\ \sup_{\substack{\xb\in\SE(\varrho)\\ \lambda^{k\top}\hat b(\xb)<\varepsilon}}\cfrac{\varepsilon-\lambda^{k\top}\hat b(\xb)}{s_k(\xb)},\label{eq:a2min}\\
\underline\alpha_g(\alpha_2):=\max_{k\in\Kp}\sigma_k^{-1}\big(\varepsilon-\inf_{\SE(\varrho)}(\lambda^{k\top}\hat b+\alpha_2 s_k)\big),\label{eq:agmin}
\end{gather}
where $\sup\varnothing:=0$.
The function $\underline\alpha_g$ is nonincreasing, since $s_k\ge0$ makes the infimum in \eqref{eq:agmin} nondecreasing in $\alpha_2$.

\emph{Rays in $\Kz$.} Here $\sigma_k=0$.
Let $D_k:=\{\xb\in\SE(\varrho):\ \lambda^{k\top}\hat b(\xb)\le\varepsilon\}$.
If $s_k(\xb)=0$ for some $\xb\in\SE(\varrho)$, then $\lambda^{k\top}\Psi(\xb)=\lambda^{k\top}H(\xb)=0$, since both terms of $s_k$ are nonnegative and $c>0$; hence $\xb\in Z_k$ and $\lambda^{k\top}\hat b(\xb)\ge\ThB_k\ge\ThB>\varepsilon$, that is, $\xb\notin D_k$.
Therefore $s_k>0$ on $D_k$.
If $D_k\neq\varnothing$, then $D_k$ is compact, $\underline s_k:=\min_{D_k}s_k>0$, and every quotient in \eqref{eq:a2min} is at most $(\varepsilon-\min_{\SE(\varrho)}\lambda^{k\top}\hat b)/\underline s_k$; hence $\underline\alpha_2$ is finite.
Since $\alpha_2>0$ and $s_k\ge0$, the inequality $\lambda^{k\top}\hat b(\xb)+\alpha_2 s_k(\xb)\ge\varepsilon$ is satisfied at every $\xb$ with $\lambda^{k\top}\hat b(\xb)\ge\varepsilon$, and at every $\xb$ with $\lambda^{k\top}\hat b(\xb)<\varepsilon$ it is equivalent to $\alpha_2\ge(\varepsilon-\lambda^{k\top}\hat b(\xb))/s_k(\xb)$.
Therefore $\nu_k(\SE(\varrho))\ge\varepsilon$ if and only if $\alpha_2$ is at least the inner supremum in \eqref{eq:a2min}, and $\min_{k\in\Kz}\nu_k(\SE(\varrho))\ge\varepsilon$ if and only if $\alpha_2\ge\underline\alpha_2$.
The strict inequality $\varepsilon<\ThB$ is required, since otherwise $s_k$ may vanish on $D_k$ and $\underline\alpha_2$ may be infinite.

\emph{Rays in $\Kp$.} Here $\sigma_k>0$, and the infimum in \eqref{eq:agmin} is finite, since $\lambda^{k\top}\hat b+\alpha_2 s_k$ is continuous on the compact set $\SE(\varrho)$.
Hence $\nu_k(\SE(\varrho))\ge\varepsilon$ if and only if $\ag\ge\big(\varepsilon-\inf_{\SE(\varrho)}(\lambda^{k\top}\hat b+\alpha_2 s_k)\big)/\sigma_k$, and the maximum over $k\in\Kp$ yields $\ag\ge\underline\alpha_g(\alpha_2)$.
Since $\min_k\nu_k(\SE(\varrho))=\min\{\min_{k\in\Kz}\nu_k(\SE(\varrho)),\ \min_{k\in\Kp}\nu_k(\SE(\varrho))\}$, the two cases prove the equivalence.

\emph{(a)} We show the two inequalities between the supremum in (a) and $\ThB$.
Let $\alpha_2,\alpha_u,\ag>0$ be arbitrary; here $\alpha_u$ need not equal $c\,\alpha_2$.
For $k\in\Kz$ and $\xb\in Z_k$, \eqref{eq:bsplit} and $\sigma_k=0$ yield $\lambda^{k\top}b(\xb)=\lambda^{k\top}\hat b(\xb)$, and therefore $\nu_k(\SE(\varrho))\le\lambda^{k\top}\hat b(\xb)$.
Taking the infimum over $\xb\in Z_k$ yields $\nu_k(\SE(\varrho))\le\ThB_k$, which is trivial if $Z_k=\varnothing$.
Hence $\min_k\nu_k(\SE(\varrho))\le\min_{k\in\Kz}\ThB_k=\ThB$ for every $(\alpha_2,\alpha_u,\ag)$, and the supremum in (a) is at most $\ThB$.

Conversely, let $\varepsilon$ be a real number with $\varepsilon<\ThB$.
Choose $\alpha_2>0$ with $\alpha_2\ge\underline\alpha_2$ and $\ag>0$ with $\ag\ge\underline\alpha_g(\alpha_2)$, both finite by the equivalence, and $\alpha_u=c\,\alpha_2>0$.
By the equivalence, $\min_k\nu_k(\SE(\varrho))\ge\varepsilon$ at this point of the orthant, so the supremum in (a) is at least $\varepsilon$.
Since $\varepsilon<\ThB$ is arbitrary, the supremum is at least $\ThB$, and it equals $+\infty$ if $\ThB=+\infty$.
The two inequalities prove (a).

\emph{(b)} Let $\alpha_1\ge\alpha_1^{\min}$, and let $Z_k(\alpha_1)$ and $E(\alpha_1)$ denote $Z_k$ and $E$ at $\alpha_1$, with $E(\alpha_1)$ computed by \arxref{eq:Eset}.
We first show that $H_k\subseteq Z_k(\alpha_1)$ for every $k\in\Kz$.
Let $\xb\in H_k$.
Since $\Cz\Ab\xb=0$, \eqref{eq:psi1} yields $\psi_1^{\pm}(\xb)=\alpha_1h_1^{\pm}(\xb)$, and the conditions $\psi_1^{\pm}(\xb)\ge0$ and $\lambda^{k\top}\Psi(\xb)=0$ become $h_1^{\pm}(\xb)\ge0$ and $\lambda^{k\top}_{z+}h_1^{+}(\xb)+\lambda^{k\top}_{z-}h_1^{-}(\xb)=0$.
These, together with $h_u^{\pm}(\xb)\ge0$ and $\lambda^{k\top}H(\xb)=0$, are independent of $\alpha_1$.
Since $\W$ grows with $\alpha_1$, \arxref{eq:Eset} yields $E(\alpha_1^{\min})\subseteq E(\alpha_1)$, which contains $\eta(\xb)$.
Thus $\xb\in Z_k(\alpha_1)$.
Next, by \eqref{eq:bsplit}, $\hat b$ depends on $\alpha_1$ only through $\mp\alpha_1\Cz\Ab\xb$, which vanishes on $H_k$.
Therefore, by \eqref{eq:ThB} and the inclusion, $\ThB(\alpha_1,\varrho)\le\min_{k\in\Kz}\inf_{\xb\in H_k}\lambda^{k\top}\hat b(\xb)=\ThH(\alpha_1^{\min},\varrho)$.

\emph{(c)} For $k\in\Kz$ let $I_k^{\pm}:=\{i:(\lambda^k_{z\pm})_i>0\}$ and $J_k^{\pm}:=\{j:(\lambda^k_{u\pm})_j>0\}$, and let $k$ satisfy $I_k^{+}\cup I_k^{-}\neq\varnothing$, $J_k^{+}\cup J_k^{-}\neq\varnothing$, $I_k^{+}\cap I_k^{-}=\varnothing$ and $J_k^{+}\cap J_k^{-}=\varnothing$ (otherwise $Z_k=\varnothing$).
Define
\begin{equation}\label{eq:zdeg}
\begin{gathered}
z^\circ_i=\begin{cases}(\zmin)_i&i\in I_k^{+}\\ (\zmax)_i&i\in I_k^{-}\\ \tfrac12((\zmax)_i+(\zmin)_i)&\text{else}\end{cases},\\[2pt]
\dot z^\circ_i=\begin{cases}(w_z)_i&i\in I_k^{+}\\ -(w_z)_i&i\in I_k^{-}\\ 0&\text{else}\end{cases},\\[2pt]
u^\circ_j=\begin{cases}\varrho(\umax)_j&j\in J_k^{+}\\ \varrho(\umin)_j&j\in J_k^{-}\end{cases},
\end{gathered}
\end{equation}
let $M_k:=\big[\Cz\T\ \ (\Cz\Ab)\T\ \ {-(\Kfull)_{J_k}\T}\big]\T$ and let $x_0,x_1$ be minimum-norm solutions of
\begin{equation}\label{eq:x01}
M_kx_0=(z^\circ;0;u^\circ),\qquad M_kx_1=(0;\dot z^\circ;0),
\end{equation}
where $(\Kfull)_{J_k}$ collects the rows $j\in J_k^{+}\cup J_k^{-}$ of $\Kfull$, and, with $d_z:=\lambda^k_{z-}-\lambda^k_{z+}$ and $d_u:=\lambda^k_{u+}-\lambda^k_{u-}$, set
\begin{equation}\label{eq:qcoef}
q_2:=d_z\T\dot z^\circ,\qquad q_l:=d_z\T\Cz\Ab^2x_l+d_u\T\Kfull\Ab\,x_l,
\end{equation}
for $l\in\{0,1\}$, and
\begin{equation}\label{eq:Aset}
A:=\{\alpha_1>0:\ x_0+\alpha_1x_1\in\SE(\varrho)\}.
\end{equation}
Solutions of \eqref{eq:x01} exist generically when $|J_k^{+}\cup J_k^{-}|\le n_p-m$, the input equations being that many affine conditions on the coordinate $\eta$ of \cref{lem:normalform}\ref{lem:nf-e}; $A$ is an interval, each constraint of $\SE(\varrho)$ being affine in $\alpha_1$ along $x_0+\alpha_1x_1$; any admissible $k$ may be used in \eqref{eq:a1bar}, and the smallest resulting $\bar\alpha_1$ is taken.
Let $k\in\Kz$, $\alpha_1>0$ and $\xb\in Z_k(\alpha_1)$.
On $\SE(\varrho)$ every summand of $\lambda^{k\top}\Psi$ is nonnegative, so $\lambda^{k\top}\Psi(\xb)=0$ gives $\psi^{+}_{1,i}(\xb)=0\ \forall i\in I_k^{+}$ and $\psi^{-}_{1,i}(\xb)=0\ \forall i\in I_k^{-}$, that is, $\dot z_i=\alpha_1h^{+}_{1,i}(\xb)\ \forall i\in I_k^{+}$ and $\dot z_i=-\alpha_1h^{-}_{1,i}(\xb)\ \forall i\in I_k^{-}$.
Substituting these into \eqref{eq:bsplit},
\begin{equation}\label{eq:bhatZk}
\begin{split}
\lambda^{k\top}\hat b(\xb)=&-\alpha_1^2\Big[(\lambda^k_{z+})\T h_1^{+}(\xb)+(\lambda^k_{z-})\T h_1^{-}(\xb)\Big]\\
&+(\lambda^k_{z-}-\lambda^k_{z+})\T\Cz\Ab^2\xb\\
&+(\lambda^k_{u+}-\lambda^k_{u-})\T\Kfull\Ab\xb
\end{split}
\end{equation}
$\forall\xb\in Z_k(\alpha_1)$, and by \eqref{eq:ThB}, $\ThB(\alpha_1,\varrho)\le\lambda^{k\top}\hat b(\xb)\ \forall k\in\Kz,\ \forall\xb\in Z_k(\alpha_1)$.
Let $k$ be admissible in \eqref{eq:zdeg}--\eqref{eq:x01}, let $\alpha_1\in A$, and set $\xb:=x_0+\alpha_1x_1$.
By \eqref{eq:x01}, $\Cz\xb=z^\circ$, $\Cz\Ab\xb=\alpha_1\dot z^\circ$ and $u_j(\xb)=u^\circ_j\ \forall j\in J_k^{+}\cup J_k^{-}$, so $\psi^{\pm}_{1,i}(\xb)=0\ \forall i\in I_k^{\pm}$ and $h^{\pm}_{u,j}(\xb)=0\ \forall j\in J_k^{\pm}$; with $\xb\in\SE(\varrho)$ by \eqref{eq:Aset}, $\xb\in Z_k(\alpha_1)$.
In \eqref{eq:bhatZk}, $h^{+}_{1,i}(\xb)=(w_z)_i\ \forall i\in I_k^{+}$, $h^{-}_{1,i}(\xb)=(w_z)_i\ \forall i\in I_k^{-}$ and $\xb=x_0+\alpha_1x_1$, so the right-hand side equals $q_2\alpha_1^2+q_1\alpha_1+q_0$ of \eqref{eq:qcoef}, with $q_2=-\sum_{i\in I_k^{+}}(\lambda^k_{z+})_i(w_z)_i-\sum_{i\in I_k^{-}}(\lambda^k_{z-})_i(w_z)_i<0$.
Therefore, for each admissible $k$, the quadratic $q$ with the coefficients \eqref{eq:qcoef} and the set $A$ in \eqref{eq:Aset} satisfy \eqref{eq:qbound} with $q_2<0$; hence $q$ and $A$ of \eqref{eq:qdef} and \eqref{eq:qbound} exist.

\emph{(d)} Since $\ThB>0$ by hypothesis, $\varepsilon=0$ satisfies $\varepsilon<\ThB$, and $\min_k\nu_k(\SE(\varrho))\ge0$ is the SIOCF condition \eqref{eq:siocfcond} on $\SE(\varrho)$.
The equivalence with $\varepsilon=0$ proves (d).
\end{proof}

\end{document}